\documentclass[cmp]{svjourmod}

\usepackage[pdftex]{graphicx}
\usepackage{amssymb}
\usepackage{bm}
\usepackage{psfrag}
\usepackage{amsmath}
\usepackage{amsfonts}
\usepackage{mathrsfs}
\usepackage{cite}

\newcommand{\p}{\partial}

\newcommand{\dd}{{\rm d}}

\begin{document}

\title{Area theorem and smoothness of compact Cauchy horizons}

\author{E. Minguzzi}
\institute{Dipartimento di Matematica e Informatica ``U. Dini'', Universit\`a degli
Studi di Firenze,  Via S. Marta 3,  I-50139 Firenze, Italy \\
\email{ettore.minguzzi@unifi.it} }
% \\ Phone: +39 055 4796 253, Fax: +39 055 471787 }
\authorrunning{E. Minguzzi}

%\author{E. Minguzzi\thanks{
%Dipartimento di Matematica e Informatica ``U. Dini'', Universit\`a
%degli Studi di Firenze, Via S. Marta 3,  I-50139 Firenze, Italy.
%E-mail: ettore.minguzzi@unifi.it} }

\date{}

\maketitle

\begin{abstract}
\noindent
We obtain an improved version of the area theorem for not necessarily differentiable horizons which, in conjunction with a recent result on the completeness of generators, allows us to prove that under the null energy condition every compactly generated Cauchy horizon is smooth and compact. We explore the consequences of this result for  time machines, topology change, black holes and cosmic censorship. For instance, it is shown that compact Cauchy horizons cannot form in a non-empty spacetime which satisfies the stable dominant energy condition wherever there is some source content.
\end{abstract}

\setcounter{tocdepth}{2}
\tableofcontents

\section{Introduction}
Many classical results of mathematical relativity have been proved under differentiability assumptions on event or Cauchy horizons.
A lot of research has been recently devoted to the removal of these conditions on horizons, particularly in the study of horizon symmetries \cite{friedrich99}  or in the generalization of the area theorem \cite{chrusciel01}.
For instance, in a recent work \cite{minguzzi14} we proved the following result
\begin{theorem} \label{one}
Let  $H^{-}(S)$ be a compactly generated past Cauchy horizon  for some  partial Cauchy hypersurface $S$, then $H^{-}(S)$ has future complete generators (and dually).
\end{theorem}
soon followed by a simplified proof by Krasnikov \cite{krasnikov14}.
This result was originally established by Hawking and Ellis \cite[Lemma 8.5.5]{hawking73} under tacit differentiability assumptions.  We mention this new version since it will play a key role in what follows.

 It has been observed that imposing strong differentiability properties, and possibly even analyticity, on the spacetime manifold, metric, or Cauchy hypersurfaces does not guarantee that the horizons will be differentiable. Indeed, an example by Budzy\'nski, Kondraki and Kr\'olak \cite{budzynski03} shows that non-differentiable compact Cauchy horizons may still form.

 In this work we wish to improve further our knowledge of general horizons, first proving a strong version of the area theorem and then showing that under the null convergence condition, which is a weak  positivity condition on the energy density,  compact Cauchy horizons are actually as regular as the metric.
 Thus, the differentiability of horizons follows from   physical conditions and has, in turn, physical consequences, most notably for topology change \cite{geroch67,chrusciel93,chrusciel94}. It is therefore reasonable to maintain that the differentiability of horizons has deep physical significance.

We recall the present status of knowledge on the differentiability properties of  horizons \cite{chrusciel98,beem98,chrusciel98b,chrusciel02}.

\begin{theorem} \label{jpf}
Let $H$ be a horizon.
 \begin{enumerate}
 \item  $H$ is differentiable  at a point $p$ if and only if $p$ belongs to just one generator.
\item  Let $\mathcal{D}$ be the subset of differentiability points of $H$. Then $H$ is $C^1$ on $\mathcal{D}$ endowed with  the induced topology.
 \item  $H$ is $C^1$ on an open set $O$ if and only if $O$ does not contain any endpoint.
 \item If $p$ is the endpoint of just one generator then $H$ is not differentiable in any neighborhood of $p$.
 \end{enumerate}
\end{theorem}

  The first (and second) statement of Theorem \ref{jpf} will be given an independent proof based solely on the semi-convexity of the horizon in Sect.\ \ref{dif}. In order to make the paper self  contained we outline the proofs of the other statements, as they are very instructive. They can be skipped on first reading.

\begin{proof}[Outline]
 The second statement as given in \cite{beem98} by Beem and Kr\'olak is weaker since they assume that the horizon is differentiable on an open set. However, the proof does not depend on this assumption and is simple: if the future-directed lightlike tangents to the generators of   $H$ (semitangents) at $\mathcal{D}$, normalized with respect to an auxiliary Riemannian metric, did not converge to that of $q\in \mathcal{D}$, then by a limit curve argument there would be at least two generators passing through $q$ in contradiction with the differentiability of $H$ at $q$.

In one direction the third statement follows from the first two, indeed if $O$ does not contain any endpoint then $H$ is differentiable in $O$ by the first statement and hence $C^1$ in $O$ by the second statement. The other direction follows from this idea: let $p\in O$, then there is a lightlike vector field $n$ tangent to $H$ which is continuous on $O$, thus by Peano existence theorem \cite[Theor.\ 2.1]{hartman64} there is a $C^1$ curve $\gamma$ contained in $H$, with tangent $n$ which passes through $p$. Since $n$ is lightlike, $\gamma$ is lightlike, and since $H$ is locally achronal, $\gamma$ is achronal and hence a geodesic. Thus it coincides with a  segment of generator of $H$ and $p$ belongs to the interior of a generator, hence it is not an endpoint.

The fourth statement follows from the third: every open neighborhood $O$ of $p$ must contain points where $H$ is not differentiable otherwise $H$ would be $C^1$ in $p$ and hence there would be no endpoint in $O$, a contradiction since $p$ is an endpoint.  $\square$
\end{proof}

Beem and Kr\'olak also obtained a further result on the differentiability properties of compact Cauchy horizons\cite[Theor.\ 4.1]{beem98}
\begin{enumerate}
\item[{\em 5}.] Let $S$ be a partial Cauchy hypersurface and assume  the null convergence condition. Assume that $H^{-}(S)$ is compact  and contains an open set $G$ such that $H^{-}(S)\backslash G$ has vanishing Lebesgue ($\mathcal{L}^n$) measure. Then $H^{-}(S)$ has no endpoint and is $C^1$.
\end{enumerate}
In this statement they  assumed the existence of the open set $G$ because they applied a flow argument by Hawking \cite[Eq.\ (8.4)]{hawking73}, which holds only for differentiable horizons. This argument was used by Hawking in some proofs \cite{hawking73,hawking92} including that of the area theorem, which has been subsequently generalized in a work by Chru{\'s}ciel  Delay,  Galloway and  Howard \cite{chrusciel01} to include the case of non-differentiable horizons. These authors briefly considered whether their area theorem could be used to remove the condition on the existence of $G$, but their analysis was inconclusive in this respect. So the problem of the smoothness of compact Cauchy horizon remained open so far.

It should be mentioned that Beem and  Kr\'olak's strategy requires  the proof of the future completeness of the generators for non-differentiable horizons. As we mentioned, we proved this result in a recent work \cite{minguzzi14}.

We shall prove the smoothness of compact Cauchy horizons under the null convergence condition without using a flow argument but rather passing through  a stronger form of the area theorem. The proof will involve quite advanced results from analysis and geometric measure theory, including recent results on the divergence theorem and regularity results for solutions to quasi-linear elliptic PDEs.

In the last section we shall apply this result to a classical problem in general relativity: can a civilization induce a local change in the space topology or create a region of chronology violation (time machine)? We shall show that under the null convergence condition both processes are impossible in classical GR. We shall also prove the classical result according to which the area of event horizons is non-decreasing, and that it increases whenever there is a change in the topology of the horizon.

We now mention the negative results on the differentiability properties of horizons. A Cauchy horizon which is $C^1$ need not be $C^2$ (cf.\ \cite{beem98}). A compact Cauchy horizon with edge  can be non-differentiable on a dense set \cite{chrusciel98,budzynski99}. If $S$ is a compact Cauchy hypersurface then a compact horizon $H^{-}(S)$ need not be $C^1$ (cf.\ \cite{budzynski03}). As the authors of the last counterexample explain:
   \begin{quote}
We have not verified whether our example fulfills
any form of energy conditions, and it could still be the case that an energy condition
together with compactness would enforce smoothness.
 \end{quote}
Our main result implies that their example must indeed violate the null energy condition.

We end this section by introducing  some definitions and terminology.
A {\em spacetime} $(M,g)$ is a paracompact, time oriented Lorentzian manifold of dimension $n+1\ge 2$. The metric has signature $(-,+,\cdots,+)$.

We  assume that  $M$ is  $C^k$, $4 \le k \le \infty$, or even analytic, and so as it is $C^1$ it has a unique $C^\infty$ compatible structure (Whitney) \cite[Theor.\ 2.9]{hirsch76}.
Thus we could assume that $M$ is smooth without loss of generality. However, one should keep in mind that due to the transformation properties of functions and tensors under changes of coordinates, a smooth tensor field  over the manifold defined by the smooth atlas would be just $C^{k-1}$ with respect to the original $C^k$ atlas, and a smooth function  for the smooth atlas would be $C^k$ for the original atlas. Thus the word {\em smooth} when used with reference to a tensor field over a possibly non-smooth manifold must always be understood in this sense.  For clarity we shall most often state the differentiability degree of the mathematical objects that we introduce.

The metric will be assumed to be $C^3$ but it is likely that the degree can be lowered. We assume this degree of differentiability because  we shall use the $C^2$ differentiability of the exponential map near the origin. As noted above the problems that we wish to solve in the next sections are present even if we assume $M$ and $g$ smooth or analytic.  By definition of time orientation $(M,g)$ admits a global smooth timelike vector field $V$ which can be assumed normalized, $g(V,V)=-1$, without loss of generality.

The {\em chronology violating set} is defined by $\mathcal{C}=\{p: p\ll p\}$, namely it is the (open) subset of $M$ of events through which passes a closed timelike curve.
A {\em lightlike line} is an achronal inextendible causal curve, hence a lightlike geodesic without conjugate points.  A  future inextendible causal curve $\gamma: [a,b)\to M$ is {\em totally future imprisoned} (or simply future imprisoned) in a compact set $K$, if there is $t_0\ni [a,b)$ such that for $t>t_0$, $\gamma(t)\in K$.
A {\em partial Cauchy hypersurface} is an acausal edgeless (and hence closed) set.
The past {\em domain of dependence} $D^{-}(S)$ of a set $S$ is the set of points $q$ for which any future inextendible causal curve starting from $q$ reaches $S$.
Observe that for a partial Cauchy hypersurface $\textrm{edge} (H^{-}(S))=\textrm{edge} (S)=\emptyset$ (cf.\ \cite[Prop.\ 6.5.2]{hawking73}). Since every generator terminates at the edge of the horizon, the generators of $H^-(S)$  are future inextendible lightlike geodesics.
The {\em null convergence condition} is: $Ric(n,n)\ge 0$ for every lightlike vector $n$. It coincides with the null energy condition under the Einstein's equation with cosmological constant (which we do not impose).

We assume the reader to be familiar with basic results on mathematical relativity \cite{hawking73,beem96}, in particular for what concerns the geometry of null hypersurfaces and horizons \cite{kupeli87,galloway00}.

%
%
%Let $(M,g)$ be a $C^{k+1}$ manifold endowed with a $C^{k}$ metric, where $k\ge 1$.
%We shall replace the atlas of $M$ with a smooth $C^k$-compatible atlas.
%Of course this means that if we find that a vector field is smooth then the vector field is just $C^k$ with respect to the original atlas. Analogously, a smooth function for the new atlas is a $C^{k+1}$ function for the old atlas. The metric $g$ is still $C^k$ for the new atlas.

\subsection{Null hypersurfaces}

A $C^2$ null  hypersurface is a hypersurface with lightlike tangent space at  each point.
%The next proposition establishes that every $C^2$ null hypersurface is  geodesically ruled, namely generated by lightlike geodesics, and locally achronal. Although this result has been widely used in the literature, I am not aware of any complete proof of the achronality part, so I provide a short proof.
The next result is well known \cite{kupeli87,galloway00}. In the `only if' direction the achronality property follows from the existence of convex neighborhoods.

\begin{theorem} \label{ppp}
Every $C^2$  hypersurface $H$ is null if and only if it is locally achronal and ruled by null geodesics.
\end{theorem}

The previous result allows one to generalize the notion of null hypersurface to the non-differentiable case as done in \cite{galloway00}.
\begin{definition}
A $C^0$ {\em future null hypersurface (past horizon\footnote{Actually in \cite{chrusciel02} these sets are called {\em future horizons}, however our choice seems appropriate. Past Cauchy horizons are future null hypersurfaces, and a black hole horizon is actually the boundary (horizon) of a past set.})} $H$ is a  locally achronal topological embedded hypersurface,  such that for every $p\in H$ and for every neighborhood $U\ni p$ in which $H$ is achronal, there exists a point $q\in J^+(p,U)\cap H$, $q\ne p$.
\end{definition}
Clearly these sets are geodesically ruled because, with reference to the definition, $q\in E^{+}(p,U)$ thus, as $U$ can be chosen   convex, there is a lightlike geodesic segment connecting $p$ to $q$. This geodesic segment stays in $H$ otherwise the achronality of $H\cap U$ on $U$ would be contradicted \cite{galloway00}. If $H$ is globally achronal then the geodesic maximizes the Lorentzian distance between any pair of its points.

In this work {\em horizon} will be a synonymous for $C^0$ {\em null hypersurface}. Observe that without further mention our horizons will be edgeless.
Due to some terminological simplifications, we shall mostly consider horizons which are   future null hypersurfaces, although all results have a time-dual version.

The following extension property will be useful so we give a detailed proof.

\begin{proposition} \label{jsp}
Let $H$ be a $C^2$ null hypersurface and let $n$ be a $C^1$ field of future-directed  lightlike vectors tangent to $H$, then on a neighborhood of any $p\in H$ we can find a $C^1$ extension of $n$ to  a future-directed lightlike vector field (denoted in the same way) which is geodesic up to parametrizations ($\nabla_n n=\kappa n$).
\end{proposition}

\begin{proof}
We can find a convex neighborhood $C$ of $p$ endowed with coordinates such that $\p_\mu$ is an orthonormal basis at $p$, $\p_0$ is future-directed timelike, $n(p)=\p_0+\p_n$, and $\{\p_i, \ i=1,\cdots n-1\}$ are tangent to $H$ at $p$. The spacelike submanifold $\Sigma:=H\cap \{q:x^n(q)=0\}$ is the graph of a $C^2$ function $x^0=f(x^1,\cdots, x^{n-1})$ and by Theor.\ \ref{ppp}, $H\cap C=\exp^C E_\Sigma$, where $E_\Sigma\subset \cup_{x\in \Sigma} T_xM$ is the vector subbundle, containing $n(p)$, which consists of the lightlike tangent vectors orthogonal to $\Sigma$. For small $\epsilon$ the graphs of the functions $f+\epsilon$  define spacelike codimension 2 submanifolds $\Sigma_\epsilon\subset \{q:x^n(q)=0\}$ to which correspond an orthogonal lightlike vector subbundle  $E_{\Sigma_\epsilon}$. Then, taking a smaller $C$ if necessary, $N_\epsilon=\exp^C E_{\Sigma_\epsilon}$ provide a foliation of $C$ by null lightlike hypersurfaces. Let $\gamma$ be a Riemannian metric for which $n$ is normalized on $C\cap H$, then the lightlike future-directed tangents to the hypersurfaces $H_\epsilon$, normalized with respect to $\gamma$, provide the searched extension of $n$. By Theor.\ \ref{ppp} the integral curves of $n$ are geodesics up to parametrization. $\square$
\end{proof}

\begin{remark}
A closed lightlike geodesic does not admit a global lightlike and {\em geodesic} tangent vector field unless it is complete \cite[Sect.\ 6.4]{hawking73}.
Similarly, a smooth compact horizon might not admit a global lightlike and {\em geodesic} tangent vector field  since its generating geodesics can accumulate on themselves \cite{kupeli87}. As we wish to include the compact Cauchy horizons in our  analysis we  shall not assume that $n$ is geodesic.
\end{remark}

\subsection{The Raychaudhuri equation} \label{ral}
Let us recall the geometrical meaning of the  Raychaudhuri equation  for lightlike geodesics \cite{galloway00}. Over a $C^2$ future null hypersurface  $H$ we consider the  vector bundle $V=TH/\!\!\sim$ obtained regarding as equivalent two vectors $X, Y\in T_pH$ such that $Y-X\propto n$. Clearly, this bundle has $n-1$-dimensional fibers. Let us denote with an overline $\overline{X}$ the equivalence class of $\sim$ containing $X$. At each $p\in H$, we introduce  a positive definite metric $h(\overline{X},\overline{Y}):=g(X,Y)$,
an endomorphism (shape operator, null Weingarten map) $b\colon V_p \to V_p$,  $\overline{X}\mapsto b(\overline{X}):= \nabla_{\overline{X}} n:=\overline{\nabla_X n}$, a second endomorphism $\overline{R}\colon V_p \to V_p$, $\overline{R}(\overline{X}):=\overline{R(X,n)n}$, where as usual
\[
R(X,Y)Z=\nabla_X\nabla_YZ-\nabla_Y\nabla_X Z-\nabla_{[X,Y]} Z,
\]
and a third endomorphism $\overline{C}\colon V_p \to V_p$, $\overline{C}(\overline{X}):=\overline{C(X,n)n}$, where $C$ is the Weyl tensor.
The definition of $b$ is well posed because $g(n,\nabla_X n)=0$ and $\nabla_n n\propto n$. The definition of $\overline{R}$ is well posed since $R(n,n)n=0$ and $g(n,R(\cdot,n)n)=0$ which implies that for every $X\in T_pM$, $R(X,n)n \in T_pH$. The endomorphisms $b, \overline{R}, \overline{C}$ are all self-adjoint with respect to $h$.

A little algebra shows that
\begin{align}
\textrm{tr} \overline{R}&=\textrm{Ric}(n,n),\\
\overline{C}&= \overline{R}-\frac{1}{n-1} \,\textrm{tr} \overline{R} \,Id ,
\end{align}
namely $\overline{C}$ is the  trace-free part of $\overline{R}$. Both endomorphisms $\overline{R}$ and $\overline{C}$ depend on $n$ at the considered point $p$ but not on the whole tangent geodesic congruence to $H$. We say that the {\em null genericity condition} is satisfied at $p\in H$ if there $\overline{R}\ne 0$. Due to \cite[Prop.\ 2.11]{beem96} this condition is equivalent to the  classical tensor condition
\[
n^\gamma n^\delta n_{[\alpha} R_{\beta] \gamma \delta [\eta} n_{\mu]}\ne 0.
\]
The  derivative $\nabla_n$, which we also denote with a prime $'$, induces a  derivative $\overline{X}':=\overline{X'}$ on sections of $V$, and hence, as usual, a derivative on endomorphisms as follows $E'(\overline{X}):=(E(\overline{X}))'-E(\overline{X}')$.
Along a generator of $H$ the null Weingarten map satisfies the Riccati equation \cite{galloway00}
\begin{equation} \label{nkv}
b'=-\overline{R}-b^2+ \kappa\, b,
\end{equation}
where $\kappa$ is defined by $\nabla_n n=\kappa\, n$.
Let us define
\begin{align*}
\theta:&=\textrm{tr} \,b, \\
\overline\sigma:&=b-\frac{1}{n-1}\, \theta\, {Id},
\end{align*}
so that $\overline{\sigma}$ is the trace-free part of $b$.  They are called {\em expansion} and {\em shear}, respectively.  It is useful to observe that if $n$ is replaced by $\Omega n$, where $\Omega$ is any $C^1$ function on $H$, then $b$ gets replaced by $\Omega b$ (thus $\theta$ by $\Omega \theta$, and $\overline{\sigma}$ by $\Omega \overline{\sigma}$), and $\overline{R}$ and $\overline{C}$ by $\Omega^2 \overline{R}$ and  $\Omega^2\overline{C}$, respectively.

Let us denote for short $\sigma^2:=\textrm{tr} \overline{\sigma}^2$. A trivial consequence of this definition is $\sigma^2\ge 0$ with equality if and only if $\overline{\sigma}=0$.

Taking the trace and the trace-free parts of  (\ref{nkv})  we obtain
\begin{align}
\frac{d }{d s}\, \theta&=-\textrm{Ric}(n,n)-\sigma^2-\frac{1}{n-1} \, \theta^2+\kappa\, \theta, \qquad (\textrm{Raychaudhuri}) \label{ray} \\
\frac{d }{d s} \, \overline{\sigma}&= -\overline{C}-(\overline{\sigma}^2-\frac{1}{n-1} \, \textrm{tr} \overline{\sigma}^2 \, {Id})-\frac{2}{n-1}\, \theta\, \overline{\sigma}+\kappa \,\overline{\sigma}, \label{she}
\end{align}
where $n=d/d s$ and the   term in parenthesis is the trace-free part of $\overline{\sigma}^2$.

\begin{remark}
If $A$ is a traceless $2 \times 2$ matrix then $A^2=\frac{1}{2}\,(\textrm{tr} A^2) Id$. Thus in the physical four dimensional spacetime case ($n=3$), the term in parenthesis in Eq.\ (\ref{she}) vanishes. Curiously this observation, present e.g.\ in \cite{schneider99},  is missed in several standard references on mathematical relativity, e.g.\ \cite{hawking73}.
\end{remark}

For a tensorial formulation of the above equations see \cite[Sect.\ 4.2]{hawking73} \cite{poisson04,kar07}.
In these last references the authors extend $n$ to a lightlike field in a neighborhood of $H$, introduce first a (projection) tensor $h^\mu_\nu=\delta^\mu_\nu+n^\mu m_\nu +m^\mu n_\nu$, where $m$ is a lightlike vector field such that $g(m,n)=-1$ on $H$, and then define
\[
\theta=g^{\mu \nu} h_\mu^\alpha h_\nu^\beta n_{\alpha ; \beta}=  h^{\alpha \beta} n_{\alpha ; \beta}= n^{ \alpha}_{;\alpha}+m^\beta n_{\beta; \alpha} n^\alpha=n^{ \alpha}_{;\alpha}-\kappa.
\]
Of course, this definition is equivalent to that  given above.

\section{The area theorem}
The area theorem appeared in Hawking and Ellis book \cite[Prop. 9.2.7 p.\ 318 and Eq.\ (8.4)]{hawking73} under tacit differentiability assumptions on the horizon.
A first proof of the area theorem without differentiability assumptions was obtained by Chru\'sciel, Delay, Galloway and Howard in \cite{chrusciel01}, where they were able to compare the areas of  two spacelike sections of the horizon in which one section stays in the future of the other section.  Unfortunately, the relationship between the area increase and the integral of the expansion  is not clarified, and the domain of integration being enclosed by  two spacelike hypersurfaces is somewhat restricted.

 In this section we wish to establish an area theorem suitable for our purposes. We shall provide a self contained and comparatively short proof of a reasonably strong version of the area theorem. This will be possible thanks to the following improvements:
\begin{itemize}
\renewcommand{\labelitemi}{$\circ$}
\item We will recognize that each global timelike vector field induces a smooth structure on the horizon. This structure can be used to integrate over the horizon  as it is usually done on open sets of $\mathbb{R}^n$. Ultimately, this approach simplifies considerably the analysis of these hypersurfaces as the  results on which we shall be interested will turn out to be independent of the smooth structure placed on the horizon.
 \item We will be able to express all the spacetime quantities of our interest (e.g. the expansion) through variables living on the horizon.  This approach will suggest immediately their generalization to the non-differentiable case.
 \item We will take advantage of the  strongest results on the divergence theorem so far developed in analysis \cite{silhavy05,chen05,chen09,pfeffer12}. This classical theorem can be improved generalizing the domain of integration (and its boundary) or generalizing the vector field. The generalization of the domain was accomplished
     by the Italian school (Caccioppoli, De Giorgi) through the introduction of domains of bounded variation or  with finite perimeter. For what concerns the vector field, it was proved
     that it does not need to be defined everywhere as long as it is Lipschitz (Federer) or even of  bounded variation. These generalizations can be applied jointly.
\end{itemize}
Our proof will differ from that of \cite{chrusciel01} since we shall use the divergence theorem while these authors  study  the sign of the Jacobian of a flow induced by the generators. However, we shall use some geometrical ideas contained in \cite{galloway00,chrusciel01} for what concerns the semi-convexity of horizons,  and the relationship between the sign of the expansion and the achronality of the horizon.

With a sufficiently strong version of the area theorem we will be able to prove      the smoothness of compact Cauchy horizons. Here the main idea is to prove that $\theta=0$ making use of the area theorem, regard this equality as a second order quasi-linear elliptic PDE (in weak sense) for the local graph function $h$ describing the horizon, and then use some well known results on the regularity of solutions to quasi-linear elliptic PDEs to infer the smoothness of the horizon.

We need to introduce some mathematical results that we shall use later on.

\subsection{Mathematical preliminaries: lower-$C^2$ functions}

Let us recall the definition of {\em lower}-$C^k$ function due to Rockafellar \cite{rockafellar82} \cite[Def.\ 10.29]{rockafellar09}.
\begin{definition}
A function $f: \Omega\to \mathbb{R}$, $\Omega\subset \mathbb{R}^n$, is lower-$C^k$, $1 \le k \le \infty$, if for every  $\bar{x}\in \Omega$ there is  some open neighborhood $O\ni \bar{x}$ and
 a representation
\begin{equation} \label{nna}
f(x) = \max_{s\in S} F(x,s), \qquad \forall x\in O,
\end{equation}
where $S$ is a compact topological space and $F\colon O \times S \to \mathbb{R}$ is a function which has partial derivatives up to order $k$ with respect to
$x$ and which along with all these derivatives is continuous not
just in $x$, but jointly in $(x,s) \in O\times S$.
\end{definition}

Clearly a lower-$C^{k+1}$ function is lower-$C^{k}$ but it turns out that
 the notions of  lower-$C^2$ and lower-$C^\infty$ function are equivalent \cite{rockafellar82}. Actually, one could introduce a notion of lower-$C^{1,1}$ function but this would be equivalent to lower-$C^2$ (see \cite{vial83,urruty85,daniilidis05}).
If $\tilde{f}-f$ is $C^k$ and $f$ is lower-$C^k$ then $\tilde{f}$ is also lower-$C^k$, it is sufficient to take $\tilde{F}(x,s)=F(x,s)+\tilde{f}(x)-f(x)$. We shall be  interested in lower-$C^2$ functions.
Convex functions are special types of lower-$C^2$ functions for which $F(x,s)$ can be chosen to be affine for every $s$ (see \cite[Theor.\ 5]{rockafellar82} or \cite[Theor.\ 10.33]{rockafellar09}). So a function which differs from a convex function by a $C^2$ function is also lower-$C^2$. Rockafellar has also proved \cite[Theor.\ 6]{rockafellar82}:

\begin{theorem} \label{axd}
For a locally Lipschitz function $f: \Omega\to \mathbb{R}$, $\Omega\subset \mathbb{R}^n$, the following properties are equivalent:
\begin{itemize}
\item[(i)] $f$ is lower-$C^2$,
\item[(ii)] for every $\bar{x}\in \Omega$ there is a convex neighborhood $O\ni \bar x$ on which $f$ has a representation, $f=g-h$, where $g$ is convex and $h$ is $C^2$,
\end{itemize}
If these cases apply, $h$ in point (ii) can be chosen quadratic and convex, even of the form $\frac{1}{2}\rho \vert x-\bar{x}\vert^2$ for some $\rho(\bar{x})$.
\end{theorem}

The last claim follows from (ii) observing that locally a $C^2$ function can be made convex by adding to it a quadratic homogeneous convex function, indeed any $C^2$ function has Hessian locally bounded from below which can be made positive definite with this operation.

Rockafellar observes \cite[Cor.\ 2]{rockafellar82} that these functions are almost everywhere twice differentiable by (ii) and Alexandrov's  theorem \cite{evans92}, namely for almost every $\bar{x}\in \Omega$, there is a quadratic form $q$ such that
\[
f(x)=q(x)+o(\vert x-\bar{x}\vert^2).
\]
Another property inherited from convex functions is that of being Lipschitz and  strict differentiable wherever they are differentiable \cite[Sect.\ 25]{rockafellar70}\cite{rockafellar82} (Peano's strong (strict) differentiability coincides with the single valuedness of Clarke's generalized gradient \cite{clarke75}).

The lower-$C^2$ property   appeared under a  variety of names in the literature, including ``weak convexity'' \cite{vial83}, ``convexity up to square'', ``generalized convexity'' and  ``semi-convexity''  \cite{andersson98,chrusciel01,cannarsa04}.
A related notion is that of proximal subgradient \cite{clarke95}.

\begin{definition}
A vector $v\in \mathbb{R}^n$ is called a {\em proximal subgradient}
of a continuous function $f \colon \Omega \to \mathbb{R}$ at $\bar{x}\in \Omega$,  if there is some neighborhood $O(\bar{x})\ni \bar{x}$ and some constant  $\rho(\bar{x})\ge 0$ such that for every  $x\in O$
\begin{equation} \label{cgf}
f(x)\ge f(\bar{x})+v \cdot (x-\bar{x})-\frac{1}{2}\rho \vert x-\bar{x}\vert^2.
\end{equation}
Whenever this inequality is satisfied $v$ is also called $\rho$-{\em proximal subgradient}.
\end{definition}
Thus the existence of a  proximal subgradient at $\bar{x}$ corresponds to the existence of a local quadratic support to $f$ at $\bar{x}$. Clearly, if $f$ is differentiable at $\bar{x}$ then its proximal subgradient coincides with $\nabla f$.
A proximal subgradient can also be characterized as follows \cite[Theor.\ 8.46]{rockafellar09}
\begin{proposition} \label{axo}
A vector $v$ is a proximal subgradient of $f$ at $\bar{x}$  if and only if on some
neighborhood of $\bar{x}$ there is a $C^2$ function $h\le f$ with $h(\bar{x}) = f(\bar{x})$ and $h'(\bar{x}) = v$.
\end{proposition}

A differentiable function can have  proximal subgradient at each point without being lower-$C^2$, e.g.\ $f=x^2 \sin (1/x)$. Observe that in this example the {\em semi-convexity constant} $\rho$ varies from point to point.

\begin{definition}
A function $f\colon \Omega \to \mathbb{R}$ is called $\rho$-lower-$C^2$ (or $\rho$-semi-convex) if there are continuous functions $c\colon S \to \mathbb{R}$, $b\colon S \to \mathbb{R}^n$, defined on some compact set $S$ such that  $f(x)=\max_{s\in S} \{-\frac{1}{2}\rho \vert x\vert^2+b(s) \cdot x+c(s)\}$.
\end{definition}

Vial \cite{vial83} and Clarke et al. \cite[Theor.\ 5.1]{clarke95} proved the following equivalence (see also \cite[Prop.\ 1.1.3]{cannarsa04}, \cite[Lemma 3.2]{andersson98})
\begin{proposition} \label{alo}
Let $\Omega\subset \mathbb{R}^n$ be open, convex and bounded and let $f\colon \Omega \to \mathbb{R}$ be Lipschitz. For any given $\rho>0$ the following properties are equivalent:
\begin{itemize}
\item[(a)] $f$ is $\rho$-lower-$C^2$,
\item[(b)] $f$  admits a $\rho$-proximal subgradient at each point, that is Eq.\ (\ref{cgf}) holds, where $\rho$ does not depend on the point $\bar{x}$,
\item[(c)] $f=g-\frac{1}{2}\, \rho \vert x\vert^2$ where $g$ is convex,
\item[(d)] for all $x,y\in \Omega$ and $\lambda \in [0,1]$
\[
(1-\lambda)f(x)+\lambda f(y)-f((1-\lambda)x+\lambda y)\ge -\rho\, \frac{\lambda(1-\lambda)}{2} \, \vert y-x\vert^2.
\]
\end{itemize}
In this case any proximal subgradient  of $f$ satisfies (\ref{cgf}) with the constant $\rho$.
\end{proposition}

For some authors the {\em semi-convex} functions are those selected by this theorem (Cannarsa and Sinestrari \cite[Chap.\ 1]{cannarsa04}). However, it must be stressed that the family of functions which locally satisfy the above proposition for some $\rho>0$, and which we call {\em semi-convex functions with locally bounded semi-convexity constant}, denoting them with $BSC_{loc}(\mathbb{R}^n)$ (reference \cite{cannarsa04} uses $SCL_{loc}(\mathbb{R}^n)$ for this family), is smaller than the family of Rockafellar's semi-convex functions $SC_{loc}(\mathbb{R}^n)$  for which the semi-convexity constant may be unbounded on compact sets as in the above example $f=x^2 \sin (1/x)$. Finally, we mention the inclusion $C^{1,1}_{loc}(\mathbb{R}^n) \subset BSC_{loc}(\mathbb{R}^n)$, cf.\ \cite{vial83}.

\subsection{Horizons are lower-$C^2$-embedded smooth manifolds} \label{kdp}

Let $V$ be a  global  future-directed normalized timelike smooth vector field.
Let $H$ be a horizon and let $p\in H$. Let $S$ be a smooth spacelike manifold transverse to $V$ at $p$, defined just in a convex neighborhood of $p$ in which $H$ is achronal.
Let $C$ be an open cylinder generated by the flow of $V$ with transverse section $S$. Let  $s\colon  S \to \mathbb{R}^{n}$ be a chart on $S$ which introduces coordinates $\{x^i\}$, $i=1,\cdots,n$ on $S$ (also denoted with ${\bf x}$).
Every $q\in C$, reads $q=\varphi_{x^0}(s^{-1}(x^i))$, where $\varphi$ is the smooth flow of $V$. This map establishes a smooth chart on $C$. The achronality of $H$ implies that  $H$ is locally a  graph, $x^0=h({\bf x})$ where $h$ is Lipschitz. We can assume that the level sets of $x^0$ are all spacelike by taking the  height of the cylinder  sufficiently small. Thus $x^0$ is a local time function. Finally, for technical reasons, we shall introduce the coordinates on $S$ in such a way that $g^{ij}$, $i,j=1,\cdots, n$, is positive definite. By continuity we can accomplish this condition taking $S$ orthogonal to $V$ at $p$, choosing coordinates so that $g_{ij}(p)=\delta_{ij}$, and taking the cylinder sufficiently small.

The horizon $H$ can be covered by a locally finite family of cylinders $C_i$ constructed as above, and whenever two cylinders intersect, the flow of $\alpha V$, for some function $\alpha>0$, establishes a smooth diffeomorphisms between open subsets of $S_i$ and $S_j$. In other words if we locally parametrize $H$ with the coordinates ${\bf x}$ on each cylinder we get a smooth atlas for $H$ which depends on the initial choice for $V$. This smooth structure coincides with that of a quotient manifold obtained from a tubular neighborhood of $H$ under the flow of $V$.

Suppose that $(\tilde{x}^0, \tilde{{\bf x}})$ is the local chart induced by a second vector field $\tilde{V}$, and suppose that the local charts related to $V$ and $\tilde{V}$ overlap. Then in the intersection we have a smooth dependence $\tilde{x}^0(x^0,{\bf x})$, $\tilde {\bf x}(x^0,{\bf x})$. We can locally parametrize $H$ with the coordinates ${\bf x}$ or with the coordinates $\tilde {\bf x}$. The change of coordinates is  $\tilde {\bf x}(h({\bf x}),{\bf x})$ which is Lipschitz. This local argument shows that changing $V$ changes the smooth atlas assigned to $H$, where the two atlases obtained in this way are just Lipschitz compatible.
  Thus, although $H$ is a Lipschitz manifold, as a set it actually admits a smooth atlas which depends on the choice of $V$. When it comes to work with $H$ it is convenient to regard it as a Lipschitz embedding of a smooth manifold $\tilde{H}$
\[
\psi\colon \tilde{H}\to M, \qquad H=\psi(\tilde{H}).
\]
We shall denote quantities living in $\tilde H$ with a tilde.
Of course, we shall be mainly concerned with results which do not really depend on the smooth structure that we have placed on $H$. It is worth to recall that the notions of Hausdorff volumes behave well under Lipschitz changes of  chart (this is the content of the change of variable formulas \cite{evans92}).

In local coordinates $\psi: {\bf x} \mapsto (h({\bf x}),{\bf x})$, thus $\psi$ has the same differentiability properties of $h$. As proved in \cite{galloway00,chrusciel01} $h$ is $SC_{loc}$ (semi-convex). We shall actually prove that it belongs to $BSC_{loc}\subsetneq SC_{loc}$.

\begin{proposition} \label{mdp}
Locally the horizon is the graph  of a  function $h\in BSC_{loc}$, thus the embedding $\psi$ is $BSC_{loc}$ hence lower-$C^2$.
\end{proposition}

\begin{proof}
Let $p\in H\cap C$, so that in the local coordinates of the open cylinder  $C$, $p=(h(\bar{{\bf x}}), \bar{{\bf x}})$ for some $\bar {\bf x}$.  Without loss of generality we can assume $C$ to be contained in a convex normal neighborhood. Here with $\exp$ we denote the exponential map on it.
Let $\bar B(\bar {\bf x}, r)$ denote a closed Euclidean-coordinate ball centered at $\bar{\bf x}$, and let $H_r$ be the portion of horizon which is a graph over it.
The set $H_r$ and its boundary $\p H_r$ are compact. For $q\in \p H_r$, let us consider the hypersurface $G_q:=C\cap \exp_q(-N(q))$ where $N(q)\subset T_qM$ is the set of future-directed null vectors at $q$. Since $N(q)$ is a smooth manifold and $\exp_q$ is $C^2$ this is actually a null $C^2$ hypersurface (except in $q$) which is achronal in $C$. Since $C$ is contained in a convex neighborhood, those $q$ for which $G_q$ intersect the fiber of some point of $\bar B(\bar {\bf x}, r/2)$ form a compact subset $K\subset \p H_r$. The hypersurface $G_q$ is transverse to  $V$ everywhere thus it is expressible as a graph  ${\bf x} \mapsto (f_q({\bf x}),{\bf x})$ over $\bar B(\bar {\bf x}, r/2)$ where $f_q$ is $C^2$. The function $f_q$ is actually $C^2$ in $(q,{\bf x})$ because $\exp: (q,v) \mapsto (q,\exp_q v)$ is  $C^2$ in $(q,v)$ with its inverse. The Hessian $ Hf_q({\bf x})$ is bounded from below for $(q,{\bf x}) \in  K\times \bar B(\bar {\bf x},r/2)$ by $-\rho I$ for some $\rho >0$. Thus at any point ${\bf x}\in \bar B(\bar {\bf x},r/2)$ we can find a proximal subgradient given by the differential of $f_q$ for some $q\in J^+((h({\bf x}), {\bf x}))\cap K$ where the constant $\rho$ does not depend on the point chosen in $\bar B(\bar {\bf x},r/2)$.
Thus by Prop.\ \ref{alo}, $h\in BSC_{loc}$ and thus it is lower-$C^2$. $\square$
\end{proof}

\subsection{Differentiability properties of the horizon} \label{dif}

Since $H$ is a smooth lower-$C^2$ embedded manifold, its differentiability properties can be readily obtained from those for (semi-)convex functions \cite{rockafellar82,alberti92,alberti94,cannarsa04}.

We have shown that $H$ can be expressed locally as a graph $x^0=h({\bf x})$ where $h$ is semi-convex. A change of coordinates ${x'^0}=x^0-u({\bf x})$ shows that we can assume $h$ convex (this change redefines the local hypersurface $S$ which will be still transverse to $V$ losing its spacelike character. As a consequence, $x^0$ is no more a local time function. Fortunately, this last property will not be important).

For a convex function $h$ the {\em subdifferential} $\p h (\bar{\bf x})$ can be defined as the set of ${\bf v}\in \mathbb{R}^n$ such that
\begin{equation} \label{suk}
h({\bf x})-h(\bar{\bf x})\ge {\bf  v} \cdot ({\bf x}-\bar{\bf x}).
\end{equation}
This notion can be further generalized to arbitrary locally Lipschitz functions but we shall content ourselves with the convex case.
Every convex function is Lipschitz thus almost everywhere differentiable.
On every neighborhood of $\bar{\bf x}$ there will be a dense set of points ${\bf x}$ where the differential $D h({\bf x})$ exists, and for any sequence
of such points converging to $\bar{\bf x}$, the corresponding sequence of
differentials will be bounded and will have cluster points.
The set of these limits
\[
\p^* h(\bar{\bf x})=\{\lim_{  {\bf x}_k\to \bar{\bf x},} D h( {\bf x}_k), h \textrm{ is differentiable at } {\bf x}_k \},
\]
is called {\em reachable gradient}.
Clarke
proved that $\p h (\bar{\bf x})$ is the convex hull of all
such possible limits \cite[Theor.\ 3.3.6]{cannarsa04}:
\begin{equation} \label{ldp}
\p h (\bar{\bf x})=\textrm{co} (\p^* h(\bar{\bf x})),
\end{equation}
and that this set is closed and compact.

A {\em semitangent} of a past horizon at $p\in H$ is a future-directed lightlike vector $n \in T_pM$ tangent to a lightlike generator.
\begin{proposition} \label{kdy}
The subdifferential of $h$ is related to the semitangents of $H$ as follows:
\[
\p h (\bar{\bf x})=\textrm{\em co}\{{\bf v}\colon \textrm{\em Ker}(\dd x^0-{\bf v}\cdot d {\bf x}) \textrm{ \em contains a  semitangent at } \bar p=(h(\bar{\bf x}), \bar{\bf x}) \}.
\]
\end{proposition}

\begin{proof}
Let ${\bf x}$ be a point of differentiability of $h$ in a neighborhood of $\bar {\bf x}$, and let ${\bf v} \cdot \dd {\bf x}$ be its differential.
The hyperplane $P$ defined by $\textrm{Ker}(-\dd x^0+{\bf v}\cdot d {\bf x})$
approximates the graph of $h$ and hence the horizon $H$ in a neighborhood of $p=(h({\bf x}),{\bf x})$. Through the point $p$ passes a lightlike generator and any future-directed semitangent $n$ at $p$ (which could be normalized so that $g(V,n) =-1$, $V=\p_0$) belongs to this hyperplane. Thus $P$ is either timelike or lightlike, but it cannot be timelike because of the achronality of $H$. Thus $P$ is lightlike and hence must be the only lightlike hyperplane containing $n$. In conclusion, wherever $h$ is differentiable its differential  ${\bf v} \cdot \dd {\bf x}$ is univocally determined by the condition ``$\textrm{ Ker}(-\dd x^0+{\bf v}\cdot d {\bf x})$  contains a  semitangent'', in particular  up to a proportionality constant there is just one semitangent. Observe that if the semitangent $n$ is normalized with $g(V,n)=-1$ then $-\dd x^0+{\bf v}\cdot d {\bf x}=g(\cdot,n)$.

By Eq.\ (\ref{ldp}) $\p^* h (\bar{\bf x})$ consists of the limits ${\bf v}$  of these vectors. So let ${\bf v}_k\to {\bf v}$ be such that there are ${\bf x}_k\to {\bf x}$, such that $\textrm{ Ker}(\dd x^0-{\bf v}_k\cdot d {\bf x})$ contains a (unique) semitangent $n_k$, $g(V,n_k)=-1$, at $(h({\bf x}_k),{\bf x}_k)$. By the limit curve theorem \cite{minguzzi07c} (or by continuity of the exponential map) the generators converge to a generator passing from or starting at $\bar p=(h (\bar{\bf x}),\bar{\bf x})$, and by continuity its semitangent must belong to $\textrm{ Ker}(\dd x^0-{\bf v}\cdot d {\bf x})$. Thus
\[
\p^* h (\bar{\bf x})\subset \{{\bf v}\colon \textrm{\em Ker}(\dd x^0-{\bf v}\cdot d {\bf x}) \textrm{ \em contains a  semitangent at } \bar p=(h(\bar{\bf x}), \bar{\bf x}) \}.
\]

Conversely, any semitangent at $\bar{p}$ determines a null hyperplane which can be written as $\textrm{ Ker}(\dd x^0-{\bf v}\cdot d {\bf x})$ for some ${\bf v}$. Let $q\in J^+(\bar{p})$ be a point in a convex neighborhood of $p$ which belongs to the generator passing though $\bar{p}$ in the direction of the semitangent, then $\textrm{ Ker}(\dd x^0-{\bf v}\cdot d {\bf x})$ is tangent to the exponential map of the past light cone at $q$ whose graph provides a $C^2$ lower  support function for $h$ and hence shows that ${\bf v}$ is a proximal subgradient. But for a convex function any proximal subgradient belongs to the subdifferential (that is if ${\bf v}$ satisfies (\ref{cgf}) then it satisfies (\ref{suk}), see \cite[Prop.\ 3.6.2]{cannarsa04}) thus ${\bf v}$ belongs indeed to the subdifferential. $\square$
\end{proof}

Let $h$ be a (semi-)convex function defined on an open set of $\mathbb{R}^n$. For every integer $0\le k  \le n$ let us define
\[
\tilde \Sigma^k=\{{\bf x}\colon dim (\p h({\bf x}))= k\}, \qquad \Sigma^k=\psi(\tilde \Sigma^k) .
\]
Here ``dim'' refers to the dimension of the affine hull of $\p h({\bf x})$, namely the dimension of the smallest affine subspace containing it \cite[Def.\ 1.5]{alberti94}. So if $\p h({\bf x})$ is a singleton then its dimension is zero. We stress that $dim (\p h({\bf x}))$ is unrelated with the number of semitangents at $p=(h({\bf x}),{\bf x})$. For instance for a light cone issued at the origin of  $n+1$ Minkowski spacetime   $dim (\p h({\bf x}))=n$, while at the origin it has an infinite number of non-proportional semitangents.

\begin{proposition} \label{aao}
Let $p=(h({\bf x}),{\bf x})$, then the dimension of the affine space spanned by the semitangents to $H$ at $p$ equals dim$(\p h({\bf x}))+1$.
\end{proposition}

If we consider just the semitangents $n\in T_pM$, $p\in H$, normalized so that $g(V,n)=-1$ then the affine space spanned by them is dim$(\p h({\bf x}))$. We shall denote by $\Sigma^k$ the subset of $H$ of points $p$ for which  the semitangents at $p$ span a $k+1$ dimensional space, and  $\Sigma:=\cup_{i\ge 1} \Sigma^i$.

\begin{proof}
The musical isomorphism $v \mapsto g(\cdot, v)$, sends a semitangent $n$ normalized with $g(n,V)=-1$ to a 1-form $-\dd x^0+{\bf v}\cdot d {\bf x}$, which can be further sent with an affine map to ${\bf v}\in \mathbb{R}^n$. The composition of these injective affine maps is affine, thus the dimension of  the affine space spanned by the semitangents to $H$ at $p$ coincides, with the dimension of the affine space spanned by the vectors ${\bf v}$ obtained in this way which, by proposition \ref{kdy} coincides with dim$(\p h({\bf x}))$. If we remove the normalization on the semitangents the dimension of the affine hull of the semitangents increases by one which proves the proposition. $\square$
\end{proof}

Every convex function has a function of bounded variation as weak derivative \cite[Theor.\ 3, Sect.\ 6.3]{evans92} thus it is worth to recall some properties of these functions. For the following notions see e.g.\ the book by Ambrosio, Fusco and Pallara \cite{ambrosio00}.
Let $V\subset\subset U$ mean `$V$ is compactly supported in $U$'.

\begin{definition}
A function $f\in L^1_{\textrm{loc}}(U)$  has {\em locally bounded variation} in $U$ if
\[
\sup\{\int_V f \,\textrm{div} \,\varphi\, d x: \varphi \in C_c^1(V,\mathbb{R}^n),\vert \varphi \vert<1\} <\infty
\]
 for each open set $V\subset\subset U$. A similar non-local version can be given where $L^1_{\textrm{loc}}(U)$ is replaced by $L^1(U)$, and  $V\subset\subset U$ is replaced by $V=U$.
 The space of functions of (locally) bounded variation is denoted $BV(U)$ (resp.\ $BV_{\textrm{loc}}(U)$).
\end{definition}

Let us recall that a real-valued Radon measure is the difference of two (positive) Radon measures, where a  Radon measure is a measure on the $\sigma$-algebra of Borel sets which is both inner regular and locally finite. In what follows by Radon measure we shall understand {\em real} Radon measure. With $\mathcal{L}^n$ or $dx$ we shall denote the Lebesgue measure on $\mathbb{R}^n$.

A structure theorem establishes that $f$ has locally bounded variation iff its distributional derivative is a finite (vector-valued) Radon measure. More precisely  \cite[Sect.\ 5.1]{evans92} \cite[Sect.\ 5.5]{pfeffer12}:

\begin{theorem} \label{str} Let $U\subset \mathbb{R}^n$ be an open set.
A  function $f \in L^1_{\textrm{loc}}(U)$ belongs to $BV_{\textrm{loc}}(U)$ iff there is a vector Radon measure $[Df]$ with polar decomposition\footnote{This notation means measure with density $\sigma$ with respect to $\Vert Df\Vert$, that is, $d [Df] =\sigma d \Vert Df\Vert$, where $\Vert Df\Vert$ is the total variation measure, and $\sigma\colon U\to \mathbb{R}^n$ is a $\Vert Df\Vert$ measurable function such that $\vert \sigma(x)\vert=1$ $\Vert Df\Vert$ a.e..} $[Df]:= \Vert Df\Vert \,\llcorner\, \sigma$  and
\[
\int_U f \, \textrm{div}\, { \varphi}\,\dd x=-\int_U { \varphi} \cdot \sigma d \Vert Df\Vert, \qquad \forall { \varphi} \in C^1_c(U,\mathbb{R}^n).
\]
\end{theorem}

This theorem shows that $[Df]$ is the distributional gradient of $f$, $\Vert Df\Vert$ is the total variation of $[Df]$,
and
$[D_i f]:=\Vert Df\Vert \,\llcorner\, \sigma^i$ are the distributional partial derivatives of $f$ (to see this set $\varphi^j=0$ for $j\ne i$).

Let $u\in [BV_{loc}(O)]^m$, $O\subset \mathbb{R}^n$. The  set $S_u$ of points where the approximate limit of $u$ does not exist is called the {\em approximate discontinuity set}. It is a Borel set with $\sigma$-finite Hausdorff  $\mathcal{H}^{n-1}$ measure (thus negligible $\mathcal{L}^n$ measure) (see \cite[Sect.\ 5.9]{evans92}).

There  is also a Borel set $J_u\subset S_u$, called {\em (approximate) Jump set} such that $\mathcal{H}^{n-1} (S_u \backslash  J_u)=0$ and where approximate right and left limits exist everywhere, that is for every $x\in J_u$ there is a unit vector $\nu (x)$ and two vectors $u^+ (x),\, u^- (x)\in\mathbb R^m$ such that, if we  denote with $B^\pm$ the half balls
\[
B^\pm (x,r,\nu) =\{y\in B(x,r):  \pm(y-x)\cdot \nu  > 0\}\, ,
\]
then for each choice of upper or lower sign
\[
\lim_{r \to 0^+} \frac{1}{\vert B^\pm (x,r,\nu) \vert} \int_{B^\pm (x,r,\nu)} \vert u(y)-u^\pm\vert\, dy=0 .
\]
These values $u^+$, $u^-$ and $\nu$ are uniquely determined at $x\in J_u$ up to permutations $(u^+,u^-,\nu) \to (u^-,u^+,-\nu)$.
The distributional derivative of a function of bounded variation is a Radon measure which admits the Lebsegue decomposition \cite{evans92}
\[
[Du]=[Du]^a+[Du]^s,
\]
where
\[
[Du]^a= \mathcal{L}^n \llcorner D u, \qquad D u\in [L^1_{loc}(O)]^n ,
\]
 is the absolutely continuous part, and $[Du]^s$ is the singular part, that is, there is a Borel set $B$ such that $\mathcal{L}^n(O\backslash B)=0 $, and such that $[Du]^s(B)=0$. Moreover, the singular part  decomposes further as the sum of Cantor and jump parts (cf.\ \cite{ambrosio00})
\begin{equation} \label{cff}
[Du]^s=[Du]^c+[Du]^j,
\end{equation}
that is $[Du]^c(A)=0$ for every Borel set $A$ such that $\mathcal{H}^{n-1}(A)<\infty$, and for any Borel set $B$
\begin{equation} \label{mod}
[Du]^j(B)=\int_{B\cap J_u} (u^+-u^-)\otimes \nu_u \,\dd \mathcal{H}^{n-1} .
\end{equation}
A function $u\in [BV_{loc}(O)]^m$ is called {\em special} if its distributional derivative does not have Cantor part $[Du]^c=0$. These functions form a subset denoted $[SBV_{loc}(O)]^m$.

We say that a Borel subset $B$ of $\mathbb{R}^n$ is countably  Lipschitz   $\mathcal{H}^d$-rectifiable if for $i=1,2,\cdots$, there exist $d$-dimensional Lipschitz submanifolds $M_i\subset \mathbb{R}^n$ such that $\mathcal{H}^d(B\backslash \cup_i M_i)=0$ (cf.\ \cite[Def.\ A.3.4]{cannarsa04}). Observe that $\cup_i M_i$ is not demanded to be a subset of $B$.

\begin{remark}
We could have defined a finer notion of  rectifiability, replacing Lipschitzness of $M_i$ with $C^2$ differentiability as in \cite{alberti92,alberti94}, however the Lipschitz version will be fine for our purposes.
\end{remark}

The following facts are well known  \cite[Theor.\ 4.1.2, Prop.\ 4.1.3, Cor.\ 4.1.13]{cannarsa04}.
\begin{theorem} For every semi-convex function $h$ defined on an open subset  $O$ of $\mathbb{R}^n$
\begin{enumerate}
\item $\tilde \Sigma^0$ consists of those points where  $h$ is  differentiable,
\item $h$ is strongly differentiable wherever it is differentiable,
\item $\tilde \Sigma^n$ is a countable set,
\item the set of non-differentiability points $\tilde \Sigma:=\cup_{i\ge 1} \tilde \Sigma^i$ coincides with the  set $S_{Dh}$ of $Dh$ regarded as a function of $[BV_{loc}(O)]^n$, and is a countably $\mathcal{H}^{n-1}$-rectifiable set. Moreover, at $\mathcal{H}^{n-1}$-a.e. $x\in J_{Dh}$  the vector $\nu (x)$ is orthogonal to the approximate tangent space to  $J_{Dh}$ at $x$.
The vector $\nu (x)$ can be chosen so that $x\mapsto \nu (x)$ is a Borel function.
\item    $\tilde \Sigma^k$ is countably $\mathcal{H}^{n-k}$-rectifiable, thus its Hausdorff dimension is at most $n-k$.
\end{enumerate}
\end{theorem}

Let us translate this result into some statements for the horizon.
\begin{theorem} \label{qdi} For every horizon $H$ on a $n+1$ dimensional spacetime:
\begin{enumerate}
\item $H$ is differentiable precisely at the points that belong to just one generator (i.e.\ on $\Sigma^0$),
\item $H$ is $C^1$ on the set $\Sigma^0$ with the induced topology,
\item there is at most a countable number of points  with the property that every future-directed lightlike half-geodesics issued from them is contained in $H$,
\item    the set $\Sigma$ of non-differentiability points of $H$  is  countably $\mathcal{H}^{n-1}$-rectifiable, thus its Hausdorff dimension is at most $n-1$. More generally, the set of points for which the span of the semitangents has dimension $k+1$ is countably $\mathcal{H}^{n-k}$-rectifiable, thus its Hausdorff dimension is at most $n-k$.
\end{enumerate}
\end{theorem}

\begin{proof}
The second statement follows from a well known property of strong differentiability: if a function is strongly differentiable on a set then it is $C^1$ there with the induced topology \cite{nijenhuis74,mikusinski78} \cite[Prop.\ 4.1.2]{cannarsa04}.
The other statements are immediate in light of Prop.\ \ref{aao}. $\square$
\end{proof}

The first two results coincide with points 1 and 2  mentioned in the Introduction. The last result appeared  in \cite{chrusciel02} along with other results on the fine differentiability properties of horizons.

\begin{theorem} \label{pdg}
If the horizon is differentiable on an open subset or, equivalently, if it has no endpoints there, then it is  $C^{1,1}_{loc}$ there.
\end{theorem}

\begin{proof}
We know that if the, say,  {\em past} horizon $H$ is differentiable on an open subset $O$ then it is $C^1$ there. In particular, as it has no endpoints on $O$, it is locally generated by lightlike geodesics and so it is locally a {\em future} horizon. The graph function $h$ belongs to $BSC_{loc}$ and $H\cap O$ when regarded as a future horizon has a local graph function which differs from $-h$ by a smooth function. Thus $-h \in BSC_{loc}$ which implies that $D h$ is locally Lipschitz \cite[Cor.\ 3.3.8]{cannarsa04}. $\square$
\end{proof}

An immediate consequence is the following.

\begin{corollary}
On the interior of the set of differentiability points, namely on $\textrm{Int}\,\Sigma^0=H\backslash \overline{\Sigma}$ the horizon is $C^{1,1}_{loc}$.
\end{corollary}

\begin{remark}
If $h^i$ is a local graph function then $D h^i\in [BV_{loc}]^n$ and  $[\textrm{Hess} \, h^i]$ is a ($n\times n$ matrix) Radon measure.
Observe that the graph function $h^i$ relative to the cylinder $C^i$ of the covering of $H$ differs,  on an open subset $\tilde O$  of $\tilde{H}$, from the graph function $h^j$ of a neighboring cylinder  by a smooth function. As a consequence, $[\textrm{Hess} \, h^i]^s(\tilde O)=0$ iff $[ \textrm{Hess} \,h^j]^s(\tilde O)=0$. Thus it makes sense to say that the
singular part of the measure $[ \textrm{Hess} \, h]$ vanishes on some open set  as this statement does not depend on the chosen graphing function.
\end{remark}

\begin{remark} \label{rel}
Every convex function over $\mathbb{R}$ for which the Hessian has no singular part is necessarily $C^1$ as its weak derivative has an absolutely continuous representative. This result does not generalize to functions in many variables, for instance $f(x,y)=\sqrt{x^2+y^2}$ is convex and its Hessian measure is absolutely continuous with respect to Lebesgue over Borel sets which do not contain the origin. But $Df$ is a function of bounded variation hence the singular part of its differential vanishes when evaluated on sets with vanishing $\mathcal{H}^{n-1}$-measure, $n=2$. The set consisting of just the origin is one such set. Thus the Hessian measure is non-singular since the singular part does not  charge neither the origin nor its complement. This example shows that the support of the singular part of the Hessian measure is not necessarily the sets of non-differentiability points for $h$.
It can be also observed that in this example  $Df$ is a function of bounded variation for which $J_{Df}$ is empty.
Finally, observe that $f$ is the graph function of a future light cone (a past horizon) with vertex at the origin of 2+1 Minkowski spacetime.
\end{remark}

\subsection{Propagation of singularities} \label{awn}
The previous theorems constrained the size of the non-differentiability set $\Sigma$ from above. In this section we  present some results on the propagation of singularities which follow from analogous result on the  theory of semi-convex functions \cite{cannarsa04}. Intuitively they constrain  the size of $\Sigma$ from below. The rest of the work does not depend on this section.

Since we have  at our disposal Theorem \ref{qdi} (point 1), Prop.\ \ref{kdy} can be improved as follows.

\begin{proposition}
The reachable gradient of $h$ is related to the semitangents of $H$ as follows:
\[
\p^* h (\bar{\bf x})= \{{\bf v}\colon \textrm{\em Ker}(\dd x^0-{\bf v}\cdot d {\bf x}) \textrm{ \em contains a  semitangent at } \bar p=(h(\bar{\bf x}), \bar{\bf x}) \}.
\]
Furthermore, it is precisely the set of extreme points of $\p h  (\bar{\bf x})$.
\end{proposition}

\begin{proof}
One inclusion is already proved  (see the proof of Prop.\ \ref{kdy}). For the converse: any semitangent determines a null hyperplane which can be written as $\textrm{ Ker}(\dd x^0-{\bf v}\cdot d {\bf x})$ for some ${\bf v}$. Consider a sequence of points $q_k=(h_k, {\bf x}_k)$ on a geodesic segment generated by the semitangent. Since these points belong to the interior of a generator, the horizon  and hence $h$ is differentiable on them so the null plane orthogonal to the semitangent is there $\textrm{ Ker}(-\dd x^0+{\bf v}_k\cdot d {\bf x})$ for some ${\bf v}_k$ which by continuity converge to ${\bf v}$ as ${\bf x}_k \to {\bf x}$. Thus ${\bf v}\in \p^* h$.

It is clear that the extreme points of $\p h  (\bar{\bf x})$ must be contained in $\p^* h (\bar{\bf x})$. Now suppose that $\bar {\bf v}\in \p^* h (\bar{\bf x})$ is not an extreme point then by the Choquet-Bishop-de Leeuw theorem there is a normalized measure $\mu$ supported on $\p^* h (\bar{\bf x})$ and not supported on just a single point such that $\bar {\bf v}=\int_{\p^*h} {\bf v} d \mu$. Let us normalize the semitangents so that $g(n,V)=-1$, then the semitangent is related to ${\bf v}$ by $-\dd x^0+{\bf v} \cdot \dd {\bf x}=g(\cdot, n)$, thus integrating $-\dd x^0+\bar {\bf v} \cdot \dd {\bf x}=g(\cdot, \int_{\p^*h}  n d \mu)$, but the term on the left-hand side is a one form which annihilates the null hyperplane determined by a semitangent while the right-hand side is a one form which annihilates a spacelike hyperplane. The contradiction proves the desired result. $\square$
\end{proof}

A semitangent $n$ at $p$ is {\em reachable} if there is a sequence of semitangents $n_k$ at some differentiability points of the horizon $p_k\in \Sigma^0$ such that $n_k \to n$ on $TM$.

It is convenient to refer the convex hull of the semitangents at $p\in H$  as the {\em subdifferential to the horizon} at $p$ (we can also call in this way the family of non-timelike hyperplanes orthogonal to one of these vectors). The map $\beta\colon {\bf v} \to n$, where $n$ is uniquely determined by $-\dd x^0+{\bf  v}\cdot d {\bf x}=g(\cdot, n)$ gives an affine bijection between the subdifferential to $h$ and the subdifferential to the horizon restricted to the normalized semitangents: $g(V,n)=-1$.  In particular, the subdifferential to the horizon is a closed convex cone. This bijection sends also the reachable gradient to the set of reachable semitangents. Thus the previous proposition states that the set of reachable semitangents coincides with the set of extreme points of the subdifferential of the horizon.

\begin{theorem} Let $H$ be a past horizon.
\begin{itemize}
\item[(a)] Suppose that at $p\in \Sigma$  the subdifferential of the horizon  does not coincide with the future causal cone at $p$, then there is a Lipschitz curve $\sigma\colon [0,\rho]\to \Sigma$, $\sigma(0)=p$, such that $\sigma(s)\ne p$ for $s\in \ ] 0,\rho]$ and with respect to a complete Riemannian metric the aperture angle of the subdifferential cone over $\sigma$ is larger than a positive constant (thus the singularity is not isolated for $p\in \Sigma^k, 1\le k\le n-1$).
\item[(b)] Suppose that  $p \in  \Sigma^k, 1\le k\le n-1$, then there is a Lipschitz map $\sigma: D\to \Sigma$, $D\subset \mathbb{R}^{n-k}$, where $D$ is a $n-k$-dimensional disk centered at the origin, such that $\sigma(0)=p$, $\sigma(D)$ possesses a tangent space at $p$ and the $\mathcal{H}^{n-k}$ density of $\sigma(D)$ at $p$ is positive.
\end{itemize}
\end{theorem}

\begin{proof}
Proof of (a). Let $p=(h({\bf x}), {\bf x})$ in the usual coordinates. Observe that the family of ${\bf v}$ such that  $n$, determined by $-\dd x^0+{\bf v}\cdot d {\bf x}=g(\cdot, n)$, is lightlike and hence normalized $g(V,n)=-1$, is a $n-1$-dimensional ellipsoid as the bijection $\beta\colon {\bf v} \to n$ is affine with affine inverse and the subset of the light cone of vectors $n$ such that $g(V,n)=-1$ is a $n-1$-dimensional ellipsoid. The set $\p h({\bf x})$ is obtained from the convex hull of the vectors ${\bf v}$ of this ellipsoid which correspond to the semitangents, thus there is a point in the boundary of  $\p h({\bf x})$  interior  to the ellipsoid if and only if not all vectors tangent to the light cone are semitangent.

From the assumption we get that the boundary of $\p h({\bf x})$ in $\mathbb{R}^n$ contains some ${\bf v}$ such that $\beta^{-1}({\bf v})$ is not a semitangent. Since $\p^* h({\bf x})$ consists of vectors which correspond to semitangents (the reachable ones) there are points in the boundary of $\p h({\bf x})$ which do not belong to $\p^* h({\bf x})$ and so the result follows from \cite[Theor.\ 4.2.2]{cannarsa04}.

Proof of (b). Recall that the {\em relative interior} of a convex set is the interior with respect to the topology of the minimal affine space containing the convex set. By Theorem \cite[Theor.\ 4.3.2, Remark 4.3.5]{cannarsa04} we need only to show that there is some point in the relative interior of $\p h({\bf x})$ which does not belong to   $\p^* h({\bf x})$, but this is obvious from the previous ellipsoid construction taking into account that $\p h({\bf x})$ has an affine hull which has dimension smaller than $n$. $\square$
\end{proof}

\subsection{Mathematical preliminaries: the divergence theorem}
Let us introduce the sets of finite perimeter \cite{giusti84,evans92,ambrosio00,hofmann07,pfeffer12}.

\begin{definition}
An $\mathcal{L}^n$-measurable subset $E\subset \mathbb{R}^n$ has {\em (locally) finite perimeter} in $U$ if $\chi_E\in BV(U)$ (resp.\ $BV_{\textrm{loc}}(U)$).
\end{definition}

An open set $E\subset \mathbb{R}^n$ has locally finite perimeter iff $\chi_E$ has locally bounded variation, namely $D\chi_E$ is a Radon measure. Following \cite{evans92} we write $\Vert \p E\Vert$ for $\Vert D\chi_E \Vert$ and call it {\em perimeter (or surface) measure}, and we write $\nu_E:=-\sigma$. Thus in a set of finite perimeter the following result holds
\[
\int_E \textrm{div} \varphi\, d x=\int_U \varphi \cdot \nu_E \, d \Vert \p E\Vert, \qquad \forall \varphi \in C^1_c(U,\mathbb{R}^n).
\]

Later we shall use the divergence theorem over domains of  finite perimeter over a horizon. Fortunately, we do not have to specify the vector field $V$ and the corresponding  smooth structure that it determines on the horizon (Sect.\ \ref{kdp}). Indeed, we observed that they are all Lipschitz equivalent and the sets of locally finite perimeter are sent to sets of locally finite perimeter  under Lipeomorphisms (locally Lipschitz homeomorphism with locally Lipschitz inverse) \cite[Sect.\ 4.7]{pfeffer12}.

The following portion of the coarea theorem  helps us to establish whether a set has finite perimeter \cite[Prop.\ 5.7.5]{pfeffer12} \cite[Sect.\ 5.5]{evans92}.

\begin{theorem} \label{kct}
Let $f\in BV_{\textrm{loc}}(U)$ then $E_t:=\{x\in U: f(x)> t\}$ has locally finite perimeter for $\mathcal{L}^1$ a.e.\ $t \in \mathbb{R}$.
\end{theorem}

As immediate consequence is the following

\begin{corollary} \label{lls}
Let $H$ be a horizon and let $\tau\colon U \to \mathbb{R}$ be a locally Lipschitz time function defined on a neighborhood of $H$. Then for almost every $t$, the sets $\{x\in H: \pm \tau(x)>t\}$ have locally finite perimeter.
\end{corollary}

This result states that for almost every $t$ the intersections of the $t$-level set of a time function with a horizon is sufficiently nice for our purposes as its measure  is locally finite (it separates the horizon in open sets of finite perimeter). Observe that we cannot claim that for any spacelike hypersurface $S$, the intersection $H\cap S$ bounds a set of finite perimeter. However, we can always build a time function $\tau$ so that $S$ becomes a level set of it (e.g.\ consider volume Cauchy time functions on $H(S)$), so that $H\cap S$ is approximated by boundaries of sets of finite perimeter.

\begin{proof}
Let us consider the plus case, the minus case being analogous. Let $p\in H$ and let $C\ni p$ be a cylinder of the covering introduced in Sect.\ \ref{kdp}, with its coordinates $(x^0,{\bf x})$.
The function $\tau\vert_C$ and $h$ are locally Lipschitz. The composition of (vector-valued) locally Lipschitz functions is locally Lipschitz thus $\tau(h({\bf x}),{\bf x})$ has locally bounded variation (one can also use the fact that the composition of a  Lipschitz function and a vector-valued function of bounded variation has bounded variation \cite[Sect.\ 3.10]{ambrosio00}).
The desired conclusion follows from Theorem \ref{kct}. A different argument could use the results on level sets of Lipschitz functions contained in \cite{alberti13}. $\square$
\end{proof}

If the time function is sufficiently smooth we can say much more (but we shall not use the next result in what follows).
A closed subset of $\mathbb{R}^n$ has {\em positive reach} if it is possible to roll a ball over its boundary \cite{federer59}. These sets have come to be known under different names, e.g. Vial-weakly convex sets \cite{vial83} or proximally smooth sets \cite{clarke75}.

\begin{theorem} \label{ljs}
Let $H$ be a past horizon and let $\tau\colon U \to \mathbb{R}$ be a $C^2$  time function with timelike gradient defined on a neighborhood of $H$. Then for  every $t$, the set $H_t:=\{x\in H: \tau(x)\le t\}\subset H$ has locally positive reach and hence has locally finite perimeter.
\end{theorem}

\begin{proof}
Since $\tau$ is a $C^1$ time function the level set $S=\tau^{-1}(t)$ is a spacelike hypersurface. Let $p\in S\cap H$. Due to the special type of differential structure placed on $H$, which can be identified with the differential structure of a manifold locally transverse to the flow of $V$,    in a neighborhood of $p$  we have a local diffeomorphism between $H$ and $S$. Thus in order to prove the claim we have only to prove it for the set $C=S\cap J^{+}(H)$ near $p\in \p C$ where we regard $C$ as a subset of $S$. Let us prove that $C$ has positive reach. We give two proofs.

The first argument is similar to Prop.\ \ref{mdp} but worked `horizontally' instead of `vertically'. We pick some $q\in J^{+}(p)\backslash\{p\}\cap H$ and consider the intersection of  its past cone with $S$. This intersection provides a $C^2$ codimension one manifold on $S$ tangent to $C$ and intersecting $C$ just on $p$ (by achronality of $H$ near $q$). We can therefore find a small closed coordinate ball of radius $r(p)>0$ entirely contained in $S\backslash C$ but for the point $p$. By a continuity argument similar to that worked out in  Prop.\ \ref{mdp} we can find a $r>0$ independent of the point in a  compact neighborhood of $p$.

As a second argument, let ${\bf x}$ be coordinates on $S$ near $p$. Let us introduce near $p$ a spacetime coordinate system in such a way that $V=\p_0$, and let $x^0=f({\bf x})$ be the graph determined by $H$ in a neighborhood of $p$. We known that $f$ is $\rho$-lower-$C^2$, furthermore since $S$ is spacelike $H$ is not tangent to it, thus in a compact neighborhood $K$ of $p$,  $m=\textrm{inf} \{\Vert \xi\Vert, \xi \in \p f(y), y\in \p C\cap K\}$ is positive where $\p f$ is the subdifferential (cf. Sect.\ \ref{dif}). Thus, by \cite[Prop.\ 4.14(ii)]{vial83} $C$ is locally Vial-weakly convex which is equivalent to say that locally it has  positive reach \cite[Prop.\ 3.5(ii)]{vial83}.
Finally, a set of locally positive reach has locally finite perimeter \cite[Theor.\ 4.2]{colombo06}. $\square$
\end{proof}

\begin{remark} One could ask whether the intersection of the horizon with the spacelike level set of the time function contains `few'  non-differentiability  points of the horizon. The answer is negative.
It is easy to construct examples of past horizons for which for some $t$, $H\cap \tau^{-1}(t)$ consists of non-differentiability points, take for instance a circle $C$ in the plane $x^0=0$ of $2+1$ Minkowski spacetime  and define $H=\p J^+(C)$. Then defined $\tau=x^0$ we have that $H\cap \tau^{-1}(0)$ consists of non-differentiability points of the horizon. Thus the regularity of the intersection of the spacelike level set with $H$ in $H$ has little to do with the presence of non-differentiability points of $H$ on that intersection.
\end{remark}

\subsubsection{Reduced and essential boundaries}

Let $E$ be a set of locally finite perimeter in $\mathbb{R}^n$.
\begin{definition}
The {\em reduced boundary} $\mathcal{F} E$ consists of points $x\in \mathbb{R}^n$ such that
\cite[Sect.\ 5.7]{evans92}
\begin{itemize}
\item[(a)] $\Vert \p E\Vert (B(x,r))>0$ for all $r>0$,
\item[(b)] $\lim_{r \to 0} \frac{1}{\Vert \p E\Vert (B(x,r))} \int_{B(x,r)} \nu_E \, d \Vert \p E\Vert=\nu_E(x)$,
\item[(c)] $\vert \nu_E(x)\vert=1$.
\end{itemize}
\end{definition}
In short the reduced boundary consists of those points for which the average minus gradient vector of $\chi_E$ coincides with itself and is normalized \cite{ambrosio00}. According to the Lebesgue-Besicovitch differentiation theorem $\Vert \p E\Vert(\mathbb{R}^n-\mathcal{F} E)=0$. The function $\nu_E\colon  \mathcal{F} E \to S^{n-1}$ is called {\em generalized exterior normal} to $E$.

Let the {\em density} of a Borel set $E$ at $x$ be defined by
\[
d(E,x):= \limsup_{r \to 0} \frac{\mathcal{L}^n(B(x,r)\cap E)}{\mathcal{L}^n(B(x,r))} ,
\]
and let $E^\delta=\{x\colon d(E,x)=\delta \}$.
The set $E^1$ is the measure theoretic interior of $E$ and $E^0$ is the measure theoretic exterior of $E$. The {\em essential (or measure theoretic) boundary} $\p^* E$ is $\mathbb{R}^n\backslash (E^0\cup E^1)$.

The two boundaries are related by
\[
 \mathcal{F}E\subset E^{1/2} \subset \p^* E\subset \p E, \quad \mathcal{H}^{n-1}(\p^* E-\mathcal{F} E)=0.
\]
Moreover, up to a set of negligible $\mathcal{H}^{n-1}$ measure, every point belongs either to $E^1$, $E^{1/2}$ or $E^0$. Furthermore, $S_{\chi_E}=\p^* E$, $J_{\chi_E}\subset E^{1/2}$, cf.\ \cite[example 3.68]{ambrosio00}.
An important structure theorem \cite[Sect.\ 5.7.3]{evans92} by De Giorgi and Federer establishes that $\Vert \p E\Vert$ is the restriction of $\mathcal{H}^{n-1}$ to $\p^* E$ (and analogously for $\mathcal{F} E$). Moreover,  up to a $\mathcal{H}^{n-1}$ negligible set $\p^* E$ is the union of countably many compact pieces of $C^1$-hypersurfaces, that is, these boundaries are rectifiable and, moreover, over these differentiable pieces $\nu_E$ coincides with the usual normal to the $C^1$ hypersurface.

The divergence (Gauss-Green) theorem will involve an integral of the measure $\mathcal{H}^{n-1}$ over the reduced boundary. However, in the boundary term one can replace $\mathcal{F} E$ with any among $J_{\chi_E}$, $E^{1/2}$ or $\p^* E$.
In general the topological boundary cannot be used because it might have a rather pathological behavior, for instance, it can have non-vanishing $\mathcal{L}^n$ measure.

A set $E$ of finite perimeter may be altered by a set of Lebesgue measure zero and
still determine the same measure-theoretic boundary $\p^*E$. In order to remove this ambiguity, let $\textrm{cl}_* E:=\p^*E\cup E^1=\mathbb{R}^n\backslash E^0$ be the set of points of density of $E$.  A set of finite perimeter is {\em normalized} if  $\textrm{cl}_* E= E$, cf.\  \cite{pfeffer12}. It is known that $\textrm{cl}_* E \subset \textrm{cl} E$ \cite[Sect.\ 4.1]{pfeffer12} and that $E$ and $\textrm{cl}_* E$ differ by a set of vanishing $\mathcal{L}^n$ measure \cite[Theor.\ 4.4.2]{pfeffer12}, thus $\textrm{cl}_* \textrm{cl}_* E=\textrm{cl}_* E$.

\subsubsection{Lipschitz domains}

 It is worth to recall the notion of Lipschitz domain, although in our application we shall use the divergence theorem on smoother domains (for the proof of the smoothness of compact Cauchy horizon) or rougher domains (in the application to Black hole horizons).

\begin{definition}
A {\em Lipschitz domain} on a smooth manifold is an open subset $D$ whose boundary $\p D$ is locally representable  as the graph of a Lipschitz function in a local atlas-compatible chart.
\end{definition}

Lipschitz domains are quite natural because for them the topological boundary $\p D$ coincides with the {\em essential boundary} $\p^{*} D$, namely the measure theoretical notion of boundary \cite[Prop.\ 4.1.2]{pfeffer12}.
Any Lipschitz domain has locally finite perimeter \cite[Prop.\ 4.5.8]{pfeffer12}.

For Lipschitz domains  the normal $\nu_D$ can be obtained using the differentiability of the local graph map almost everywhere \cite{hofmann07}. Thus if $\p D$ is the graph of a Lipschitz function $\varphi\colon O \to \mathbb{R}$, $O\subset \mathbb{R}^{n-1}$, in some $\mathbb{R}^n$-isometric local coordinates
then the outward unit normal near $(x_0,\varphi(x_0))$ has the usual expression in terms of $\nabla \varphi$
\[
\nu_D(x,\varphi(x))= \frac{(\nabla \varphi(x),-1)}{\sqrt{1+\vert \nabla \varphi(x) \vert^2}}
\]
for a.e.\ $x$ near $x_0$, where the Euclidean area element is $\dd S=\sqrt{1+\vert \nabla \varphi(x) \vert^2} \dd x$.

\subsubsection{Divergence measure field}

For the next notion see \cite{chen01,chen05,silhavy05}.

\begin{definition}
A vector field $v\in L^1_{\textrm{loc}}(U,\mathbb{R}^n)$ is said to be a {\em divergence measure field} in $U$ if there is a Radon measure $\mu$ such that
\[
\int_U v\cdot D \varphi\,  dx=-\int_U \varphi \,d \mu, \quad \forall \varphi\in C^1_c (U)
\]
in which case we define $[\textrm{div} v]=\mu$.
\end{definition}

A vector field in $L^1_{\textrm{loc}}(U,\mathbb{R}^n)$  whose components belong to $BV_{\textrm{loc}}(U)$ is a divergence measure field (use the fact that the distributional partial derivative of $v^i$ is $[D_j v^i]$ (Theor.\ \ref{str}) in $\int_U v\cdot D \varphi\,  dx=\sum_i \int_U v^i D_i \varphi\,  dx =-\sum_i \int_U \varphi d [D_i v^i]$), see also the stronger result \cite[Prop.\ 3.4]{chen01}.

In what follows we shall be interested in divergence measure fields of bounded variation which  belong to $L^\infty_{\textrm{loc}}(U,\mathbb{R}^n)$ where $n+1$ is the spacetime dimension. As a consequence, due to the general properties of functions of bounded variation, the measure $\textrm{div} v$ will be absolutely continuous with respect to $\mathcal{H}^{n-1}$ (this fact follows from the decomposition (\ref{cff}), see also \cite[Theor.\ 3.2]{silhavy05}). These fields are {\em dominated} in {\v S}ilhav{\'y}'s terminology \cite[p.\ 24]{silhavy05}, a fact that will simplify the definition of trace that we shall give in a moment.

\subsubsection{The divergence theorem}
Let $v\in [BV_{\textrm{loc}}(U)]^n\cap L^\infty_{\textrm{loc}}(U,\mathbb{R}^n)$.
Let $a_n$ be the volume of the $n$-dimensional Euclidean unit ball. We define a function $S_v\colon \mathbb{R}^n\times S^{n-1}\to \mathbb{R}$ by
\[
S_v(x,\nu):=\lim_{r\to 0} \frac{n}{a_{n-1} r^n} \int_{B^-(x,r,\nu)} v(y) \cdot \frac{x-y}{\vert x-y\vert }\, d y,
\]
if the limit exists and is finite, and 0 otherwise.
This is a generalization of the scalar product $v(x) \cdot \nu$. Wherever  $v$ is continuous $S_v(x,\nu)=v(x) \cdot \nu$. However, more generally one should be careful because the equality $S_v(x,\nu)=-S_v(x,-\nu)$ holds only if $[\textrm{div} v]$ has no singular part on $\p^*E$ \cite[Eq.\ 4.2]{silhavy05}.

Federer has shown that the divergence theorem holds for domains with locally finite perimeter provided the vector field is Lipschitz \cite[Theor.\ 6.5.4]{pfeffer12}.
In what follows we shall use the next stronger version recently obtained by {\v S}ilhav{\'y} \cite[Theor.\ 4.4(i)]{silhavy05} and Chen, Torres and Ziemer \cite[Theor.\ 1]{chen05}, \cite[Theor.\ 5.2]{chen09}.
\begin{theorem} \label{hxo}
Let $v\in [BV_{\textrm{loc}}(U)]^n\cap L^\infty_{\textrm{loc}}(U,\mathbb{R}^n)$, let $[\textrm{div}\, v]$ be the divergence measure, and let $\varphi$ be  a locally Lipschitz function with compact support, then for every normalized set of locally finite perimeter $E$
\begin{equation} \label{ldd}
\int_E \varphi \, d [\textrm{div}\, v]+\int_E D\varphi \cdot v \, d x=\int_{\p^* E} \varphi \, S_v(x,\nu_E) \, d \mathcal{H}^{n-1}.
\end{equation}
\end{theorem}
The right-hand side, regarded as a functional on $Lip_c(U)$, is called {\em normal trace}.

Observe that if $E$ has compact closure then $\varphi$ need not have compact support since we can modify it just outside $\bar{E}$ to make it of compact support.

\begin{remark} \label{mmm}
The first integral on the left-hand side of Eq.\ (\ref{ldd}) can be further split into two terms thanks to the decomposition of $[\textrm{div}\, v]$ in a component absolutely continuous with respect to $\mathcal{L}^n$ and a singular component. If we are given a set $E$ which is not normalized then we can apply the divergence theorem to $cl_* E$, then the first term in the mentioned splitting is $\int_{cl_* E} \varphi \, d [\textrm{div}\, v]^a$ and we can replace $cl_* E$ by $E$ since these sets are equivalent in the  $\mathcal{L}^n$ measure.
\end{remark}

\subsection{Volume and area} \label{voa}
Let us suppose that $H$ is $C^2$ and let $n$ be a $C^1$ future-directed lightlike vector field tangent to it.
 We define the  volume over $H$ as the measure $\mu_H$ defined by
\begin{equation} \label{kka}
-g(n,V)\, \mu_H= i_V\mu_M
\end{equation}
where $\mu_M$ is the volume $n+1$-form on spacetime and $\mu_H$ is evaluated just on the tangent space to $H$. This choice of volume is independent of the transverse field $V$ but it depends on $n$, namely on the scale of $n$ over different generators. It is indeed impossible to give a unique natural notion of volume for $H$. This is not so for its smooth transverse sections which have an area measured by the form $\sigma$
\[
-g(n,V) \,\sigma= i_n i_V\mu_M
\]
which is independent of both $n$ and $V$ when the form $\sigma$  is evaluated on the tangent space to the section.

\begin{remark} \label{lll}
The section is a codimension 2 submanifold which belongs to a second local horizon $H'$ with semitangent $n'$. Since the corresponding forms on the section, $\sigma$ and $\sigma'$, do not depend no the choice of $V$ we can take $V=n+n'$, from which we obtain $\sigma=-\sigma'$.
\end{remark}

\begin{remark}
Introduce on the $C^2$ horizon $H$ a function $s$ which measures the integral parameter of the flow lines of $n$ starting from some local transverse section to $H$. Then on each flow line $n=\frac{d}{d s}$, and the volume reads $\mu_H=\dd s\wedge \sigma$.
\end{remark}

In the non-smooth case we place on $\tilde{H}$ a measure which is related to Eq.\ (\ref{kka})
\begin{equation} \label{aaz}
\tilde\mu_H:=\frac{\sqrt{-\vert g\vert (h({\bf x}), {\bf x})}}{-g(n,V)} \, \dd x^1\wedge \cdots\wedge \dd x^{n},
\end{equation}
where $\vert g\vert$ denotes the determinant of $g$.
Clearly, $\tilde \mu_H$ is absolutely continuous with respect to $\mathcal{L}^n$ and conversely. The measure $\mu_H$ is the push-forward of $\tilde\mu_H$ by $\psi$.

Let us find a local expression for $n$ on the differentiability set $\Sigma^0$. The form $g(n, \cdot )$ has the same kernel of $\dd (x^0-h(\bf{x}))$ thus they are proportional, the proportionality constant being fixed using  $V=\p_0$. Thus  in the local coordinates of the cylinder
\begin{equation} \label{enn}
n=g(n,V)\, g^{-1}(\cdot, \dd (x^0- h({\bf x}))),
 \end{equation}
 and the expression of the field in local coordinates is then
\[
n=[-g(n,V)]\{[g^{ij}(h({\bf x}), {\bf x}) \p_{j} h-g^{i0}(h({\bf x}), {\bf x}) ]\p_i+[g^{0j}(h({\bf x}), {\bf x}) \p_j h-g^{00}(h({\bf x}), {\bf x})] \p_0\}.
\]
The function $g(n,V)$ is arbitrary and serves to fix the scale of $n$.
The coefficients $g^{ij}, g^{i0}, g^{00}$ are Lipschitz because $h$ is Lipschitz. The degree of differentiability  of this field is the same as that of the partial derivatives $\p_j h$.

Since $\psi$ is strongly differentiable \cite{nijenhuis74} over $\tilde\Sigma^0$ we can pull back $n$ to this set (this is simply a projection to the quotient manifold $\tilde H$).
\begin{equation}
\tilde{n}=\psi^{-1}_* n= [-g(n,V)]\, [g^{ij}(h({\bf x}), {\bf x}) \p_{j} h-g^{i0}(h({\bf x}), {\bf x}) ]\p_i.
\end{equation}
The pull-backed generators are integral curves of this field.

However, we can say more on the vector field on $\tilde H$ defined through  the previous equation.
Since $h$ is locally Lipschitz we have \cite[p.131]{evans92}, $h\in W^{1,\infty}_{loc}(U,\mathbb{R})$, $U\subset \mathbb{R}^{n}$, thus $D h$ exist almost everywhere, coincides with the weak derivative almost everywhere \cite[p.232]{evans92} and belongs to $L^\infty_{loc}(U,\mathbb{R}^{n})$.
It has been proved in Sect.\ \ref{kdp} that $h$ is lower-$C^2$ (semi-convex), and since the gradient of a convex function is a function of locally bounded variation \cite[Sect.\ 6.3]{evans92} we conclude that $\tilde n\colon U\to \mathbb{R}^{n}$ belongs to  $L^\infty_{loc}(U,\mathbb{R}^{n}) \cap [BV_{loc}(U)]^n$. The differentiability properties of this vector field are rather weak but, fortunately, they meet exactly the requirements of the divergence theorem \ref{hxo}.

In what follows we  apply the divergence theorem to the vector field $n$ on $H$. It is sufficient to prove it for domains contained in the cylinders covering $H$, so we shall apply the divergence theorem for vector fields on $\mathbb{R}^{n}$. However, we have first to make sense of the divergence of $n$ using ingredients which live in $\tilde{H}$ rather than on spacetime.

The following result will be used as a guide to the non-smooth case.

\begin{proposition} \label{jso}
Let $V$ be a global smooth future-directed timelike vector field.
Let $H$ be a $C^2$ null hypersurface and  let $n$ be a $C^1$  lightlike future-directed field tangent to it. Let $p\in H$ and let us denote in the same way a lightlike pregeodesic extension of $n$ to a neighborhood $U\ni p$ as in Prop.\ \ref{jsp} (thus $\nabla_n n=\kappa\, n$ for some function $k$ on $U$).
Introduce on a neighborhood of $p$ local coordinates as done above using the flow of $V$, and regard $H$ as a local graph of a $(C^2)$ function $h$, then for every $C^1$ function $\varphi\colon M\to \mathbb{R}$
\begin{align}
\p_i\Big(\varphi \, \tilde{n}^i\frac{\sqrt{-\vert g\vert(h({\bf x}), {\bf x})}}{-g(n,V)}\Big)
%&=\p_\mu \Big[\varphi\, n^\mu \frac{\sqrt{-\vert g\vert(x^0, {\bf x})}}{-g(n,V)} \, \Big]\vert_{x^0=h({\bf x})} \nonumber \\
&= [\varphi \,\theta+\p_n \varphi]  \frac{\sqrt{-\vert g\vert(h({\bf x}), {\bf x})}}{-g(n,V)} \label{mkc}
\end{align}
where
\begin{equation} \label{the}
\theta=n^{ \mu}_{;\mu} \vert_{x^0=h({\bf x})}-\kappa
\end{equation}
and where $g(n,V)$, on the divergence in the left-hand side,  is regarded as a function on $\tilde{H}$ and hence expressed as a function of ${\bf x}$.
\end{proposition}

\begin{remark} \label{odd}
It is interesting to note the following property of $\theta$ as given by Eq.\ (\ref{the}). Rescaling $n$ as follows $\hat{n}=\Omega n$, redefines $\kappa$ as $\hat{\kappa}=\Omega \kappa+\p_n \Omega$, and finally $\hat\theta=\hat n^\mu_{; \mu}-\hat \kappa$ is related to $\theta$ by a simple rescaling: $\hat{\theta}=\Omega \theta$. In particular if $s, \hat s$ are local functions on $H$ such that over the integral curves of  $n$, $n=d/ds$, $\hat n=d/d \hat s$, then the integral elements $\theta d s=\hat \theta d \hat s$ coincide.
\end{remark}

\begin{proof}
Let $F(x^0, {\bf x})= \varphi(x^0,{\bf x}) \sqrt{-\vert g\vert(x^0, {\bf x})}/[-g(n,V)(x^0,{\bf x})]$.
We have
\begin{align}
\p_i[\tilde{n}^i F(h({\bf x}), {\bf x})]&=[\p_i h\, \p_0 (n^iF(x^0, {\bf x}))+\p_i(n^i F(x^0, {\bf x}))]\vert_{x^0=h({\bf x})} \nonumber\\
&= \{ \p_\mu(n^\mu F(x^0, {\bf x}))+ [\p_i h \,\p_0 (n^iF(x^0, {\bf x}))& \nonumber\\
&\quad -\p_0 (n^0F(x^0, {\bf x})) ]\}\vert_{x^0=h({\bf x})} \nonumber \\
&=[\p_\mu(n^\mu F(x^0, {\bf x}) )-\p_\mu(x^0-h) \,\p_0(n^\mu F(x^0, {\bf x}))]\vert_{x^0=h({\bf x})}\nonumber\\
&=[\p_\mu(n^\mu F(x^0, {\bf x}))-\p_0( F(x^0, {\bf x})\, n^\mu \p_\mu(x^0-h) )\nonumber \\
&\quad +(\p_0\p_\mu(x^0-h)) n^\mu  F(x^0, {\bf x})]\vert_{x^0=h({\bf x})}. \label{moa}
\end{align}
The last term vanishes because $h$ does not depend on $x^0$ thus  $\p_0\p_\mu(x^0-h)=0$. The penultimate term on the right-hand side can be rearranged as follows
\begin{align*}
-\p_0( F(x^0, {\bf x})\, n^\mu \p_\mu(x^0-h) ) \vert_{x^0=h({\bf x})}&= -(\p_0 F)\, \p_n(x^0-h)  \vert_{x^0=h({\bf x})}\\ & \qquad - F\, \p_V \p_n(x^0-h) \vert_{x^0=h({\bf x})}
\end{align*}
The first term vanishes because $x^0-h=0$ on $H$ and $n$ is tangent to it, so we are left with
\begin{align*}
-\p_0( F(x^0, {\bf x})\, n^\mu \p_\mu(x^0-h) ) \vert_{x^0=h({\bf x})}&=  \{- F\, [V,n] (x^0-h)  - F\,  \p_n \p_0(x^0-h) \}\vert_{x^0=h({\bf x})}.
\end{align*}
where we used $\p_0=\p_V$. The last term vanishes because $h$ depends only on ${\bf x}$. Recalling Eq.\ (\ref{enn})
\begin{align*}
-\p_0( F(x^0, {\bf x})\, n^\mu \p_\mu(x^0-h) ) \vert_{x^0=h({\bf x})}&=-Fd(x^0-h) [L_V n] \vert_{x^0=h({\bf x})}\\& =\frac{\varphi \sqrt{-\vert g\vert}}{(-g(n,V))^2} \, g(n, L_Vn)
\end{align*}
Plugging back into Eq.\ (\ref{moa}) we obtain
\begin{align*}
\p_i[\tilde{n}^i F(h({\bf x}), {\bf x})]&=\p_\mu(\varphi\, n^\mu \frac{\sqrt{-\vert g\vert}}{-g(n,V)})+\frac{\varphi \sqrt{-\vert g\vert}}{(-g(n,V))^2} \, g(n, L_Vn)\\&= \frac{\varphi }{-g(n,V)}\,\p_\mu(n^\mu  \sqrt{-\vert g\vert}) +\varphi \frac{\sqrt{-\vert g\vert}}{g(n,V)^2} \,[\p_n g(n,V)+ g(n, L_Vn)]\\& \qquad+ \p_n \varphi  \frac{\sqrt{-\vert g\vert}}{-g(n,V)} .
\end{align*}
But
\begin{align*}
\p_ng(n,V)&=g(\nabla_n n, V)+g(n,\nabla_nV)=\kappa g(n,V)- g(n,L_V n)+ g(n,\nabla_V n)\\&=\kappa g(n,V)- g(n,L_V n),
\end{align*}
where we used  $g(n,n)=0$. Finally, using $\p_\mu(n^\mu  \sqrt{-\vert g\vert})=n^\mu_{;\mu} \sqrt{-\vert g\vert}$, we obtain the desired equation. $\square$
\end{proof}

We already know from  Section \ref{ral} that the expansion is a property which depends only on the vector field $n$ over $H$ and not on its extension.
Equation (\ref{mkc}) allows us to express  $\theta$  from  quantities living in $H$. Indeed, recalling the expression for $\tilde{n}^i$ we obtain
\begin{corollary}
In local coordinates constructed in Sect.\ \ref{kdp} the expansion of a $C^2$ horizon reads
\begin{equation} \label{kks}
\theta=\frac{-g(n,V)}{\sqrt{-\vert g\vert(h({\bf x}), {\bf x})}} \,\p_i\big\{  [g^{ij}(h({\bf x}), {\bf x}) \p_{j} h-g^{i0}(h({\bf x}), {\bf x}) ] \sqrt{-\vert g\vert(h({\bf x}), {\bf x})}\,\big\}.
\end{equation}
\end{corollary}
This expression confirms that, apart from a normalizing factor dependent on the normalization of $n$, the expansion is independent of the extension of $n$ outside $H$ and can be entirely calculated in term of $h$ and its first and second derivatives.

Let us still suppose that $H$ is $C^2$ and let $\tilde D=\psi^{-1}(D)$ be a domain on $\tilde{H}$ with $C^1$ boundary $\p \tilde D$. Let $\varphi$ be a $C^1$ function in a neighborhood of $H$. The divergence theorem reads (we put a tilde whenever we wish to stress that the actual calculation is performed in $\tilde H$ but remove it whenever we want more readable equations)
\begin{equation} \label{aad}
\int_D [\varphi \, \theta  + \p_n \varphi ] \,\mu_H =\int_{\p D} \varphi \, \sigma,
\end{equation}
where
\begin{equation} \label{nps}
\tilde{\sigma}= \frac{\sqrt{-\vert g\vert(h({\bf x}), {\bf x})}}{-g(n,V)}\, i_{\tilde{n}} \, \dd x^1\wedge \cdots\wedge \dd x^{n}=i_{\tilde{n}} \tilde \mu_H,
\end{equation}
is the area form over $\p \tilde D$.
Equation (\ref{aad}) follows immediately once we regard $\tilde D$ as the union of domains $\tilde D_i$ with piecewise $C^1$ boundary such that $\psi(\tilde D_i)\subset C_i$, where $C_i$ is a cylinder of the locally finite covering of $H$. In fact, the divergence theorem must be proved only inside each subdomain $\tilde D_i$ and there it is reduced to the usual divergence theorem on $\mathbb{R}^{n}$ due to Prop.\ \ref{jso} and Eq.\ (\ref{aaz}).
\begin{align}
v^i&:= \big[g^{ij}(h({\bf x}), {\bf x}) \p_{j} h-g^{i0}(h({\bf x}), {\bf x}) \big] \sqrt{-\vert g\vert(h({\bf x}), {\bf x})} , \label{kwd}\\
\int_{\tilde D_i}  \!\!\p_i (\varphi\,v^i) \, \dd x \!\!& =\int_{\p \tilde D_i}\!\!  \varphi\, v\cdot  \nu_{\tilde D_i} \,\dd S  ,\label{kwp}
\end{align}
where $\dd x=\dd x^1\cdots \dd x^{n}$, $\nu_{\tilde D_i}$ is the normal to $\p \tilde D_i$  and  $\dd S$ is the  Euclidean area element  of $\p \tilde D_i$.

\subsection{General area theorem and  compact Cauchy horizons} %\label{kxp}

In the previous section we have obtained the divergence theorem assuming that the horizon is $C^2$. In this section we wish to remove this assumption.

Since $\p_j h$ is a function of bounded variation its derivative is a signed Radon measure
\[
\mu_{ij}:=[\p_i\p_j h]
\]
denoted on Sect.\ \ref{dif} by $[\textrm{Hess} h]$. By the Lebesgue decomposition theorem \cite[Sect.\ 1.6.2]{evans92} the measure decomposes in a measure absolutely continuous with respect to  $\mathcal{L}^{n}$ (and hence $\mu_H$) and a singular measure $[\p_i\p_j h]= [\p_i\p_j h]^a+ [\p_i\p_j h]^{s}$.
We recall that since $[\p_i\p_j h]^s$ is singular, there is a Borel set $B$, such that $\mathcal{L}^n(\tilde{H}\backslash B)=0 $, and such that $[\p_i\p_j h]^s(B)=0$.
As $D h$ has bounded variation by the Calder\'on-Zygmund theorem it is approximately differentiable almost everywhere, moreover the approximate differential coincides with $[\p_i\p_j h]^a$, cf.\ \cite[Theor.\ 3.83]{ambrosio00} and coincides with the Alexandrov Hessian \cite[p.\ 242]{evans92}. Furthermore, almost everywhere $Dh$ exists and the subdifferential admits a first order expansion (Mignot's theorem) which involves again $[\p_i\p_j h]^a$, cf.\ \cite{bianchi96,rockafellar99}.

\begin{remark} \label{pos}
For $h$ convex the Hessian measures $\mu_{ij}$ are  non-negative in the sense that for every semi-positive definite metric field $g^{ij}$ the measure $\mu_{ij} g^{ij}$ is non-negative (this is a simple improvement over \cite[Theor.\ 2, Sect.\ 6.3]{evans92} obtained replacing $\varphi$ for $g^{ij}$ in the first steps of that proof, see also \cite{reshetnjak68,dudley77}). As a consequence the same is true for $\mu_{ij}^s g^{ij}$  (evaluate it on subsets of $\tilde{H}\backslash B$) and, since $\mu_{ij}^s$ does not change if we alter $h$ by a smooth function, the non-negativity of $\mu_{ij}^s g^{ij}$ is also true for $h$ semi-convex. Next suppose that $\mu_{ij}^s g^{ij}$ vanishes and that $g^{ij}$ is positive definite. Since $g^{ij}$ is positive definite, locally we can find some $\epsilon >0$ such that $g^{ij}-\epsilon \delta_{ij}$ is positive definite, thus $\mu_{ij}^s\delta^{ij}$ vanishes. But since the trace is the sum of the (non-negative) eigenvalues, $\mu^s_{ij}$ is absolutely continuous with respect to $\textrm{tr} \mu^s_{ij}$ which implies that $\mu^s_{ij}$ itself vanishes \cite[Sect.\ 9]{dudley77}.
\end{remark}

Fortunately, equation (\ref{kwp}) still holds  when it is understood in the sense of Theorem \ref{hxo}. Indeed, since $\p_i h \in BV_{\textrm{loc}}(U,\mathbb{R}^n)\cap L^\infty_{\textrm{loc}}(U,\mathbb{R}^n)$, we also have that $v^i$ as given in Eq.\ (\ref{kwd}), satisfies  $v^i \in BV_{\textrm{loc}}(U,\mathbb{R}^n)\cap L^\infty_{\textrm{loc}}(U,\mathbb{R}^n)$,
 and so $v^i$ meets the conditions for the application of the divergence theorem.

We now define $\theta\in L^1(U,\mathbb{R})$  through the absolutely continuous part of the divergence of $v^i$, so as to recover (\ref{kks}) in the $C^2$ case
\[
[\p_i v^i] \llcorner  \frac{-g(n,V)}{\sqrt{-\vert g\vert(h({\bf x}), {\bf x})}} =\mathcal{L}^n \llcorner \,\theta +\mu_{ij}^s \llcorner \,[-g(n,V) g^{ij}] ,
\]
which can be rewritten
\begin{equation} \label{mkd}
[\p_i v^i]=\tilde{\mu}_H \llcorner \,\theta + \mu_{ij}^s \llcorner (\sqrt{-\vert g\vert}\,  g^{ij})\vert_{\tilde H} ,
\end{equation}

\begin{remark}
One can ask whether the definition of $\theta$ is intrinsic to $H$, that is,  independent of the vector field $V$ and the various coordinate constructions behind its definition. The answer is affirmative because by Alexandrov's theorem $h$ is twice differentiable almost everywhere. Let $p\in H$, $p=(h(\bar {\bf x}),\bar {\bf x}) $ where $\bar {\bf x}$ is an  Alexandrov point of $h$; let $C$ be a convex neighborhood of $p$, and let $T$ be a timelike hypersurface passing through $p$. By the Alexandrov theorem the set $H\cap T$ has second order contact with a $C^2$ codimension 2 manifold $S$ near $p$, then the expansion $\theta$  coincides with that of $E^{+}(S,C)$  (which is a $C^2$ submanifold near $p$ by the properties of the exponential map) since both depend in the same way on the second order expansion of  $h$ at $\bar{\bf x}$. The expansion of the  $C^2$ hypersurface $E^{+}(S,C)$ is also obtained from the usual intrinsic definition of Sect.\ \ref{ral} and so the expansion of $H$ at the Alexandrov points is  well posed almost everywhere and so is the function $\theta\in L^1_{loc}$.

Similarly we could have given a coordinate expression for the shear, and shown that it was well defined through an analogous argument. However, such expression will not be required. In fact, some next PDE arguments will just require the coordinate expression of $\theta$.
\end{remark}

Let us consider a set of finite perimeter $D\subset H$, and a Lipschitz function on $H$. By definition the horizon (oriented) {\em area functional} is
\begin{equation} \label{njx}
A(\p^*D,\varphi) :=\int_{\p^* \tilde D} \varphi \,S_v(x,\nu_{\tilde D}) \, \dd \mathcal{H}^{n-1} .
\end{equation}
By equation (\ref{kwp}) this is the area integral $\int_{\p D} \varphi \, \sigma=0$ when this boundary is $C^1$ and the vector field $n$ is continuous, thus the previous expression  is the measure theoretic generalization of the area integral of $\varphi$. For $\varphi=1$ the oriented area functional is the oriented area. There is some abuse of notation in Eq.\ (\ref{njx}) since $A$ depends on $D$ as well since this set determines the orientation of the normal $\nu_{\tilde D}$.

We can split the essential boundary in three pieces $\p_+ D$, $\p_-D$ and $\p_0 D$ depending on the value of $S_v(x,\nu_D)$, respectively positive, negative or zero. We call $\p_+ D$ the {\em future essential boundary} and $\p_- D$ the {\em past essential boundary}. Clearly,
\begin{equation} \label{njf}
A(\p^*D,\varphi) :=\vert A(\p_+D,\varphi)\vert- \vert A(\p_-D,\varphi)\vert,
\end{equation}
where the former term on the right-hand side represents the contribution from the boundary to the future of $D$ and the latter term represents the contribution from the boundary to the past of $D$.

The following propositions simplifies the interpretation of the boundary terms in some special cases of physical interest.

\begin{proposition} \label{alp}
Let $\tau: U\to \mathbb{R}$ be a locally Lipschitz time function defined in a neighborhood of $H$. For almost every $t$, the intersection $\tau^{-1}(t)\cap H\cap \Sigma$
is a set of zero $\mathcal{H}^{n-1}$ measure, and if $D$ is bounded by $\tau^{-1}(t)\cap H$  (e.g.\ because it is the portion of $H$ inside $\tau^{-1}([a,b])$ for some $a,b$) then $v\cdot \nu_D$ can replace $S_v(x,\nu_{D})$ in Eq.\ (\ref{njx}).
\end{proposition}

In this case since $\p^*D$ is rectifiable the area functional is the sum of contributions obtained from the classical area $\int \varphi \sigma$. Moreover, in this case $S_v(x,\nu_{D})=-S_{v}(x,\nu_{H\backslash D})$ as it follows using the expression with the scalar product. So the area integral of a surface can be calculated taking as reference the domain on one side or the complementary domain on the other side (see also \cite[Eq.\ (4.2)]{silhavy05}).

\begin{proof}
Recall that locally $\tau(h(\bf{x}), {\bf x})$ is Lipschitz.
The set of non-differentiability (or non-$C^1$) points $\Sigma$ has zero $\mathcal{L}^n$ measure.
As a consequence of the coarea formula for Lipschitz functions \cite[Sect.\ 3.4, Lemma 2]{evans98}, $\mathcal{H}^{n-1}(\Sigma\cap \tau^{-1}(t))=0$ for almost every $t$.
Thus for almost every $t$, $v$ is $C^0$ $\mathcal{H}^{n-1}$-a.e. on $\tau^{-1}(t)\cap H$, and hence $v\cdot \nu_D$ can replace $S_v(x,\nu_{D})$ on Eq.\ (\ref{njx}). $\square$
\end{proof}

Typically $D$ will bounded by two hypersurfaces transverse to the vector field $n$ on the horizon and one hypersurface tangent to it.

\begin{proposition}
Suppose that the topological boundary $\p D$ includes a $C^1$ hypersurface  $W$ such that each point of $W$ is internal to some  generator of $H$, then $W\subset \p_0 D$, that is, this portion of boundary does not  contribute to the area functional.
\end{proposition}

\begin{proof}
Since the points of $W$ are internal to some generator the horizon is $C^1$ there on the topology of the differentiability set $\Sigma^0$ (Theorem \ref{qdi} or \cite{chrusciel98,beem98,chrusciel98b}), which has full measure on any neighborhood of $p\in W$.
Thus the vector field $v$ which enters the integral of $S_v$ is continuous, which implies that $S_v(x,\nu)=v\cdot \nu$ there. Furthermore, since $W$ is $C^1$ $\dd S= \dd \mathcal{H}^{n-1}$ thus the contribution of $W$ is $\int_{W}\!\!  \varphi\, v\cdot  \nu \,\dd S$ which has been shown to be equal to $\int_{W} \varphi \, \sigma$. But this integral vanishes since $\sigma=i_{n} \mu_H$ and the tangent space at $q\in W$ includes $n$. $\square$
\end{proof}

We are ready to prove:

\begin{theorem}[Area theorem] \label{are}
Let $H$ be a past horizon, $D\subset H$ an open relatively compact subset of finite perimeter,  $\varphi$  a positive Lipschitz function on $D$, $\mu_H$ the volume form on $H$ induced by a smooth future-directed timelike vector field $V$,
 $n$  a field of semitangents normalized through the definition of an arbitrary function $-g(n,V)>0$ (hence locally given by Eq.\ (\ref{kks}))
and let $\theta\in L^1_{\textrm{loc}}(H,\mu_H)$ be such that $\mu_H  \llcorner\, \theta$ is the absolutely continuous part of the expansion (measure) of the field $n$, then
\begin{equation} \label{dik}
\int_D [\varphi \, \theta  + \p_n \varphi ] \,\mu_H \le \vert A(\p_+D,\varphi)\vert- \vert A(\p_-D,\varphi)\vert
\end{equation}
with equality if and only if the horizon is $W^{2,1}$ on $D$ (i.e.\  the local graphing function $h$ is $W^{2,1}$ and $\mu^s_{ij}=0$ on $\tilde D$).
\end{theorem}
We shall be mostly interested on this result for $\varphi=1$, for which the terms on the right-hand side are the areas of the past and future boundaries of $D$.

%\begin{remark}
%For $\varphi=1$ this result improves  the area theorem of \cite{chrusciel01}. The domain $D$ being a general open set of finite perimeter is quite general and includes sets bounded by level sets of Lipschitz time functions (provided time functions exist, see Cor.\ \ref{lls}). This wider applicability will be useful in the proof of the smoothness of compact Cauchy horizons. The definition of area (and area functional) is simplified thanks to results from geometric measure theory. The left hand side does not vanish so our formula relates the increase in area with the integral of the expansion and the smoothness of the horizon.
%In \cite{chrusciel01} the smoothness of the horizon is obtained from the equality $\vert A(\p_+D,\varphi)\vert= \vert A(\p_-D,\varphi)\vert$ under the assumption $\theta\ge 0$ while we can infer that the horizon is $C^2$  from a more general equality case.
%\end{remark}

\begin{proof}
Let $D\subset H$ be an open relatively compact set, and let us split it into the union of measurable sets with $D_k$ with  piecewise smooth boundary, but for the part they have in common with $\p D$, in such a way that $D_k\subset C_k$, where $C_k$ is the cylinder covering of $H$. This result can be accomplished cutting $H$ with a finite number of timelike hypersurfaces generated by $V$.
 We can slightly move these hypersurfaces and hence the internal boundaries of $\p D_k$ in such a way that\footnote{We have added an index $k$ to stress the dependence of some quantities on the subdomain $\tilde D_k$, $\tilde D=\cup_k \tilde D_k$.}  $\mu_{ij}^{ s \,(k)}(\p D_k\backslash \p D)=0$  (this fact follows from Fubini's theorem or the coarea formula and from the fact that  $\mu_{ij}^{ s \,(k)}$ and $\mathcal{L}^n$ are singular). Let us apply the divergence theorem to each set $cl_* \tilde D_k$.

Then the divergence theorem applied to each domain $cl_* \tilde D_k$ gives
\begin{align}
\sum_k\int_{cl_* \tilde D_k} \!\!\! (\p_i \varphi)\, v^{i} \,\dd x &+\sum_k\int_{cl_* \tilde D_k} \!\!\! \varphi\,\theta \, \dd \tilde{\mu}_H +  \sum_k\int_{cl_* \tilde D_k} \!\!\! \varphi(\sqrt{-\vert g\vert}\,  g^{ij})\vert_{\tilde H} \, \dd \mu^{ s \,(k)}_{ij} \nonumber \\
&=\sum_k\int_{\p^* \tilde D_k} \varphi\,S_v(x,\nu_{\tilde D_k}) \, \dd \mathcal{H}^{n-1}, \label{mkb}
\end{align}
where $\varphi$ is any positive Lipschitz function. On the internal boundaries the identity $S_v(x,\nu)=-S_v(x,-\nu)$ holds true since the singular part of the measure does not charge these sets, and hence each internal boundary term coming from the divergence theorem is canceled by the corresponding term relative to the domain on the other side. Thus the  right-hand side of Eq.\ (\ref{mkb}) is given by the area functional (\ref{njx}) which can be written as in Eq.\ (\ref{njf}).

We are almost done. Let us define for $D\subset H$
\[
P(D,\varphi)=\sum_k \int_{cl_* \tilde D_k} \varphi (\sqrt{-\vert g\vert}\,  g^{ij})\vert_{\tilde H} \,  \dd \mu^{ s \,(k)}_{ij} .
\]
By Remark \ref{pos} $P(D)\ge 0$ with equality if and only if $\mu^s_{ij}=0$ on $cl_* D$. By Remark \ref{mmm} in  the first two integrals on the left-hand side of  Eq.\ (\ref{mkb}) we can replace $cl_* \tilde D_k$ with $\tilde D_k$ which concludes the proof. $\square$
\end{proof}

The area theorem can be given a refined formulation with an interesting equality case as follows.

\begin{theorem}[Area theorem II] \label{lla}
Under the assumption of the previous theorem
\begin{equation} \label{dih}
\int_D [\varphi \, \theta  + \p_n \varphi ] \,\mu_H+2\int_{\Sigma^1(H)\cap D} \varphi \,\sigma \le \vert A(\p_+D,\varphi)\vert- \vert A(\p_-D,\varphi)\vert ,
\end{equation}
with equality if and only if the semitangent field is a local special (vector) function of bounded variation (that is,  $\p h \in [SBV_{loc}(O)]^n$ where $h$ is the local graphing function).
\end{theorem}

Observe that for $\varphi=1$ the second term on the left-hand side $2\sigma(\Sigma^1(H)\cap D)$ is twice the (non-negative) area of the set of non-differentiability points $\Sigma^1$.

\begin{proof}
We know that on each open set $D^{(k)}$,  $\mu_{ij}^{ s \,(k)}=[\p_i\p_j h^{(k)}]^s=[\p_i\p_j h^{(k)}]^j+[\p_i\p_j h ^{(k)}]^c$ where the measures on the right-hand side are non-negative since they are mutually singular. Thus the proof goes as before where this time we include the jump term in Eq.\ (\ref{mkb}) as given by Eq.\ (\ref{mod}). We have
\begin{align*}
\int_{cl_* \tilde D_k} \!\!\!\!\!\! \!\!\!\varphi(\sqrt{-\vert g\vert}\,  g^{ij})\vert_{\tilde H} \, \dd \mu^{ j \,(k)}_{ij}&=\int_{cl_* \tilde D_k\cap J_{Dh}^{(k)}} \!\!\!  \!\!\!\!\!\!\!\!\!\!\!\!\!\!\! \varphi(\sqrt{-\vert g\vert}\,  g^{ij})\vert_{\tilde H} \, (Dh^+-Dh^-)_i (\nu_{Dh^{(k)}})_j \, d \mathcal{H}^{n-1}\\
&=\int_{cl_* \tilde D_k\cap \tilde \Sigma^1} \!\!\!\!\!\!\!\!\! \varphi \vert_{\tilde H} \, (v^+-v^-)^i (\nu_{Dh^{(k)}})_i \, d \mathcal{H}^{n-1} ,
\end{align*}
where we used Theorem \ref{qdi} to express the domain in terms of $\tilde \Sigma^1$.  Recall that $\tilde \Sigma^1$ is a countably $C^2$  $\mathcal{H}^{n-1}$-rectifiable \cite{alberti92,alberti94}, thus we can replace $d \mathcal{H}^{n-1}$ with the usual area of the rectifying hypersurface, and the last term of the previously displayed equation can be recast, as done for Eqs.\ (\ref{nps}) and (\ref{kwp}) as twice the integral of $\varphi$ in the area of $\Sigma^1$ (the two contributions from $v^+$ and $v^-$ are the same due to Remark \ref{lll}). Finally, by construction the singular measure does not charge $\p^*\tilde D_k\cap \tilde D$, thus $cl_* \tilde D_k$ can be replaced with $\tilde D_k$. $\square$
\end{proof}

\begin{example}
Let us consider a 2+1 Minkowski spacetime of coordinates $(t,x,y)$ and metric $g=-\dd t^2+\dd x^2+\dd y^2$, and let us define $M$ removing from it the circle of radius $1$ on the plane $t=2$ with center $(2,0,0)$. Let us consider the  circle $C$ of radius $3$ in the plane $t=0$ with center at the origin, and let $H=E^+(C)$.
We have $\Sigma^1=C$. Let $\varphi=1$, $D=t^{-1}((-1,1))\cap H$, $V=\p_0$, then $\sigma(\Sigma^1\cap D)=6\pi$, $\vert A(\p_+D,1)\vert=12\pi$, $\vert A(\p_-D,1)\vert=0$, $\theta=\pm (x^2+y^2)^{-1/2}$ with the plus (minus) sign on the portion of horizon outside (resp.\ inside) $C$. Then Eq.\ (\ref{dik}) holds with the equality sign and indeed it is clear that the distributional Hessian of $h(x,y)= \vert(x^2+y^2)^{1/2}-3\vert$ has no Cantor part since $h$ is $C^2$ outside $\tilde \Sigma^1$.
\end{example}

Of course,  it is important to establish when $\theta\ge 0$.
The next result has been proved in \cite[Lemma 4.2]{galloway00} \cite{chrusciel01}. Since our assumptions are slightly different we provide a proof.

\begin{theorem} \label{pot}
Suppose that the null convergence condition holds. Let $H$ be an achronal past horizon whose generators are future complete, then $\theta\ge 0$, $\mu_H$-(and $\mathcal{L}^n$-)almost everywhere.
\end{theorem}

Actually, we shall need achronality of $H$ on just a neighborhood of it, still this local achronality property is slightly stronger that that  included in the definition of $C^0$ null hypersurface.

\begin{proof}
Almost every point of $\tilde{H}$ is an Alexandrov point for the lower-$C^2$ function $h$. Let ${\bf x}\in \tilde{H}$ be an Alexandrov point for which $\theta$ as given in Eq.\ (\ref{kks}) is negative.
 Let $S\subset \tilde{H}$ be a hypersurface  transverse to $\tilde{n}$ and passing through ${\bf x}$. Since $h$ has second order expansion at ${\bf x}$ we can define a quadratic function $\tilde{h}$ on $S$ whose graph is tangent to that of $h$ at ${\bf x}$ and stays above $h$ (quadratic upper support) since it has larger Hessian. On spacetime the graph of $\tilde{h}$ and the boundary of its causal future define a null hypersurface $N$  which is tangent to $H$ at $p=(h({\bf x}),{\bf x})$, is $C^2$ near $p$, has in common with $H$ the lightlike generator $\gamma$ passing through $p$ and, if  $\tilde{h}$ is chosen sufficiently close to $h$, has an expansion which is negative at $p$ (because it depends on the linear and quadratic terms in the Taylor expansion of $\tilde{h}$, see Eq.\ (\ref{kks}) and they are chosen to approximate those of $h$). Thus by a standard argument which uses the completeness of $\gamma$, $N$ develops a focusing point to the future of $p$ on $\gamma$. But since $N$ stays to the future of $H$ near $p$ it would follow that $H$ is not achronal, a contradiction. $\square$
\end{proof}

The following very nice result will not be used but is really worth to mention. The proof can be found in \cite[Theor.\ 5.1]{chrusciel01} so we just sketch its main idea.

\begin{theorem} \label{abp}
Let $p$ be an Alexandrov point for the past horizon $H$, and let $\gamma\colon [0, a] \to H$, $t \mapsto \gamma(t)$, $\gamma(0)=p$ be a (segment of) generator. Then any point in $\gamma$ is an Alexandrov point. Moreover, the Weingarten map $b$ is continuously differentiable over $\gamma$ and satisfies the optical equation (\ref{nkv}). In particular, the Raychaudhuri equation (\ref{ray}) and the evolution equation for the shear hold on $\gamma$.
\end{theorem}

\begin{proof}[Sketch of proof]
As in the proof of Theorem \ref{pot} let $S\subset \tilde{H}$ be be a hypersurface  transverse to $\tilde{n}$ and passing through ${\bf x}$. Since $h$ has second order expansion at ${\bf x}$ we can define {\em two} quadratic functions $h^\pm$ on $S$ whose graphs are tangent to that of $h$ at ${\bf x}$ and stay above $h$ in the plus case (quadratic upper support, larger Hessian) or below it (quadratic lower support, smaller Hessian). On spacetime the graphs of ${h}^\pm$ define two local condimension 2 submanifolds $\sigma^\pm$ passing through $p=(h({\bf x}),{\bf x})$ which do not intersect $I^{-}(\gamma(a))$.

The boundary of $J^{+}(\sigma^{\pm})$  defines a null hypersurface $N^\pm$  which is tangent to $H$ at $p$, is $C^3$ near $p$ and,  has in common with $H$  the segment of generator $\gamma$.
Since $N^\pm$ are $C^3$ they satisfy the optical equation (\ref{nkv}). Thus the sections of $N^\pm$ have quadratic approximation on $\gamma$ determined by the Weingarten map $b^\pm$.
We can take a succession $h^\pm_n \to h$ at $p$ and so obtain hypersurfaces $H^\pm_n$ and through evolution a sequence of maps $b^\pm_n$ defined on $[0,a]$. Since the optical equation depends continuously on the initial condition the maps $b^\pm_n$ converge, at any point of $\gamma$ to the evolution $b_N(t)$ of $b_N(0)$ as calculated using the Alexandrov Hessian at $p$. But as locally $N\subset J^+(N^-_n)\cap J^{-}(N^+_n)$, this quadratic support bound implies that  $N$ admits quadratic Taylor expansion and that its Weingarten map is indeed $b_N(t)$. $\square$
\end{proof}

The next result
proves that the conditions $\mu^s_{ij}=0$ and $\theta=0$ force the horizon to be smooth. Observe that if the horizon is $C^1_{loc}$ then by Theor.\ \ref{pdg} it is $C^{1,1}_{loc}$ which implies that $\mu^s_{ij}=0$, however the converse does not hold:  $\mu^s_{ij}=0$ does not imply that the horizon is $C^1_{loc}$, see Remark \ref{rel}.

\begin{theorem} \label{app}
 Suppose that a past horizon has local graphing function $h$ with vanishing singular Hessian part, i.e.\ $\mu^s_{ij}=0$ (for instance the horizon is $C^1_{loc}$). If  the non-singular Hessian part
 satisfies $\theta=0$ on an open set $O$, then the horizon has at least the same regularity as the metric, thus smooth if the metric is smooth, and analytic if the metric is analytic.
 More generally, if $\theta({\bf x})$ does not necessarily vanish but is locally bounded   the horizon is $C^{1,1}_{loc}$ in $O$.
\end{theorem}

It is worth to recall that every $C^1$ manifold admits a unique smooth compatible structure (Whitney) and any smooth manifold admits a unique compatible analytic structure (Grauert and Morrey). Thus there is no ambiguity on the smooth or analytic structure placed on the horizon.

\begin{proof}

Let $p\in O$ and let $k$ be such that $p\in D_k=\psi(\tilde D_k)$. Let $O_k:=D_k\cap O$.
By assumption $h$ has first weak derivative in $L^\infty_{loc}(\tilde O_k)$ ($h$ is Lipschitz) and  second order weak derivatives in $L^1_{loc}(\tilde O_k)$ (because $\mu_{ij}^s=0$). In particular $h\in W^{1,2}_{loc}(\tilde O_k)$.

As $\mu_{ij}^s=0$  Eq.\ (\ref{mkd}) implies that locally $[\p_i v^i]=\tilde{\mu}_H \llcorner \,\theta$, thus  by the divergence theorem for every $\varphi \in C^\infty_c(\tilde O_k, \mathbb{R})$
\[
\int_{\tilde O_k} v^i \p_i \varphi \,d x=- \int_{\tilde O_k}  \varphi \,\theta \,\tilde{\mu}_H.
\]
By Eq.\ (\ref{kwd}) this is a quasi-linear elliptic differential equation (recall that $g^{ij}$ is positive definite in the coordinate cylinder) in divergence form for which   $h$ is a weak solution
\begin{equation} \label{mld}
\p_i (G^{ij}({\bf x},h) \p_{j} h-f^i({\bf x},h) )= \frac{\sqrt{-\vert g\vert(h({\bf x}), {\bf x})}}{-g(n,V)} \,\theta({\bf x})=0 ,
\end{equation}
where
\begin{align*}
G^{ij}({\bf x},h)&= g^{ij}(h, {\bf x})  \sqrt{-\vert g\vert(h, {\bf x})},\\
f^i({\bf x},h)&= g^{i0}(h, {\bf x}) \sqrt{-\vert g\vert(h, {\bf x})},
\end{align*}
have the same degree of differentiability of the metric. As the metric is $C^l$, $l\ge 3$,
they are $C^{l-1,\alpha}$ for some  $\alpha\in (0,1]$. Observe that $h$ is locally bounded and $G^{ij}$ satisfies a uniform ellipticity condition (i.e.\ its eigenvalues are locally bounded by positive constants from above and from below). By a well known result by Lady{\v{z}}enskaja and Ural$'$tseva  which generalizes De Giorgi regularity theorem \cite[Theor.\ 6.4, Chap.\ 4]{ladyzhenskaya68} (apply it with $m=2$, $a=0$) the weak solution $h$ is actually $C^{l, \alpha}$, thus smooth if the metric is smooth. The more general statement with  $\theta$  locally bounded  follows observing that this condition implies that the right-hand side of Eq.\ (\ref{mld})  $\frac{\sqrt{-\vert g\vert(h({\bf x}), {\bf x})}}{-g(n,V)} \,\theta({\bf x})$ is a locally bounded function of ${\bf x}$ (recall that $-g(n,V)>0$ can be  chosen arbitrarily, Sect.\ \ref{voa}). Thus by a general result by Tolksdorf \cite{tolksdorf84} on quasi-linear PDEs,  the weak solution $h$ is $C^1_{loc}$, and so the horizon in $C^{1,1}_{loc}$ by Theor.\ \ref{pdg} (for the sake of comparison with the literature we stress that if $h$ were vector valued then its regularity could be assured only up to a set of measure zero as first observed by De Giorgi \cite{giaquinta83}).
A similar result by Petrowsky and Morrey \cite{morrey58} proves that the solution is analytic if the coefficients of the quasi-linear equation are analytic.  $\square$
\end{proof}

We are ready to state our main theorem which establishes that under a rather weak positive energy condition the compact Cauchy horizons are smooth. The proof  uses Theorem \ref{one} on the completeness of generators of compactly generated Cauchy horizons.

\begin{theorem} \label{mai}
Let  $S$ be a connected partial Cauchy hypersurface and suppose that the null convergence condition holds. If $H$  is a compactly
generated component of $H^{-}(S)$ then it coincides with $H^{-}(S)$, it  is compact\footnote{The result that the compactly generated horizons are actually compact has been first  obtained in \cite{hawking92} and  \cite[Theor. 12]{budzynski01} under smoothness assumptions on the horizon.} and  $C^{3}$. Actually smooth if the metric is smooth, and analytic if the metric is analytic. Moreover,  $S$ is  compact with zero Euler characteristic, $H$  is generated by future complete lightlike lines and on $H$
\[\theta=\sigma^2=Ric(n,n)=0,\qquad b=\overline{\sigma}=\overline{R}=\overline{C}=0.\]
In other words, for every $X\in TH$, $\nabla_X n \propto n$ and $R(X, n)n\propto n$, that is, the second fundamental form vanishes on $H$ and the null genericity condition is violated everywhere on $H$.
\end{theorem}

Some comments are in order. Any closed manifold of odd dimension has zero Euler characteristic, so for the physical four dimensional spacetime case $(n=3)$ there is no need to write ``with zero Euler characteristic'' in the above statement.
Without the connectedness condition on $S$ we cannot infer that $H$ coincides with $H^{-}(S)$, so it can be removed if it is known that  the whole $H^{-}(S)$ is compactly generated.
The null convergence condition is necessary for without this assumption Budzy\'nski, Kondracki and Kr\'olak have been able to construct an example of compact  Cauchy horizon  which has no edge and is not differentiable \cite{budzynski03}.

\begin{proof} By Theorem \ref{one} the generators are future complete, and by Theorem \ref{pot} $\theta\ge 0$ $\mathcal{L}^n$-almost everywhere.
Let $K$ be the compact set in which all generators of $H$ are future imprisoned.
Since $H\cap K$ is compact $K$ is covered by a finite number of coordinate cylinders. We can replace $K$ by the union of the closure of these cylinders, thus $K$ can  be chosen such that $\tilde{D}:=\psi^{-1}(H\cap K)\subset \tilde{H}$ has piecewise $C^1$ boundary and hence has finite perimeter.
Since $H\cap K$ is compact its $\mu_H$-measure is finite. Let $D=H\cap K$
 then the right-hand side of Eq. (\ref{dik}) is non-positive since $\vert A(\p_+D,\varphi)\vert=0$ (no generator escapes $D$ so $S_v<0$). But the left-hand side of (\ref{dik}) is non-negative thus both sides are zero.
 As we have equality $\theta=0$ and $\mu_{ij}^s=0$
  by Theorem \ref{are},  hence the horizon is smooth by Theorem \ref{app}, and generated by lightlike lines (Theorem \ref{jpf}).
Moreover, $H$ must be entirely contained in $K$ for otherwise the generators entering $K$ would imply $ \vert A(\p_-D, 1)\vert>0$ and hence a negative right-hand side. As a consequence, $H$ is compact. Any global timelike past-directed vector field, when suitably normalized, has a 1-flow map which establishes a homeomorphism between $H$ and a subset $W$ of $S$. But $H$ has no edge thus $W$ cannot have boundary on $S$, so as $S$ is connected, $W=S$ and hence $H=H^{-}(S)$.
As $H$ is compact with zero Euler characteristic (it admits a $C^0$ field of semitangents), $S$ has the same properties.
The Raychaudhuri equation (\ref{ray})
and the null energy condition imply $\sigma^2=Ric(n,n)=\textrm{tr} \overline{R}=0$ and hence $\overline{\sigma}=0$ on $H$. The evolution equation for the shear (\ref{she}) implies $\overline{C}=0$ and hence $\overline{R}=0$. Since $\theta=0$ and $\overline{\sigma}=0$ we have also $b=0$. $\square$
\end{proof}

\section{Applications}

In this section we explore some applications of the area theorem and the smoothness of compact Cauchy horizons.

\subsection{Time machines}
Hawking's classical theorem on chronology protection \cite{hawking92} is:
\begin{theorem}
Let $(M,g)$ be a spacetime which satisfies the null convergence condition. Let $S$ be a non-compact connected partial Cauchy hypersurface, then $H^+(S)$ cannot be compactly generated.
\end{theorem}
This theorem is contained in
 Theorem \ref{mai} (in the time dual version).

The original proof by Hawking, based on an area argument similar to  that employed in \cite[Eq.\ (8.4)]{hawking73}, was incomplete as he assumed a $C^2$ horizon at two steps: in order to claim that the generators are complete and in order to apply his flow argument to the horizon. Our Theorem \ref{mai} jointly with Theorem \ref{one} solves the problems of Hawking's original argument and proves, furthermore, that every compact Cauchy horizon is smooth.

According to Hawking  this theorem implies that regions of chronology violation (time machines) cannot form starting from nice initial conditions. The argument is as follows: on a spacetime admitting a non-compact Cauchy hypersurface $S$  the construction of a time machine (e.g.\ a region of chronology violation) by some advanced  civilization  would necessarily imply the formation of   a horizon $H^{+}(S)$ which, being originated by the actions of that civilization on a limited spacetime region, i.e.\ a compact set, would have its past generators entering that region.  In other words, $H^+(S)$ would have to be compactly generated.  The theorem proves that there is a contradiction, hence the time machine cannot form.

A more precise result forbidding the formation of time machines will be obtained in Theorem \ref{laz}.

\subsection{Topology change}

The question of topology change in general relativity has attracted considerable interest \cite{geroch67,tipler77,chrusciel93,borde94b,borde97}. Reinhart  \cite{reinhart63,geroch67} proved that in four spacetime dimensions any two spacelike 3-manifolds can be connected by a Lorentzian cobordism, namely the Lorentzian condition on the metric does not restrict the possibilities of topology change.
 Geroch was able to prove that any topology change which takes place over a compact region implies the formation of closed timelike curves, so he showed that  chronology forbids topology change (of course there could  be  topology change in non-compact regions as the breaking of the spacetime continuum might lead to any sort of phenomena \cite{yodzis73,krasnikov95}).

Tipler  gave two theorems in which he removed the chronology condition by imposing the weak energy condition and the null genericity condition \cite[Theor.\ 4,6]{tipler77}. These theorems, which relied on the usual differentiability assumptions on horizons, have been considered by many people the last word on the subject.
Unfortunately, physically speaking,  the genericity condition might even less justified than  chronology. Indeed, the genericity condition makes physical sense only over lightlike geodesics that are not imprisoned in a compact set, namely only in those cases in which it is known that they probe a {\em non-finite} region of spacetime.
%Hawking and Penrose had this case in mind in their famous paper on singularities  where they introduced the concept of null genericity \cite{hawking70}.

With Theorem \ref{mai} we have seen that the very presence of compact Cauchy horizons implies a violation of the null genericity condition. Under tacit differentiability  assumptions on the horizon this result could be found in \cite{borde84}.
However, we have no physical reasons to be as confident in the validity of the null genericity condition on compact subsets as to exclude the formation of compact Cauchy horizons. Thus we cannot assume the null genericity condition in the study of topology change as done by Tipler \cite{tipler77} and Borde \cite[Theor.\ 3]{borde94b}.

However, thanks to the result on the smoothness of horizons we can solve the problem of topology change in a satisfactorily way by imposing just the null convergence condition.

\begin{theorem} \label{laz}
Let $(M,g)$ be a spacetime which satisfies the null convergence condition.
Suppose that two disjoint $C^1$ spacelike hypersurfaces are Lorentz cobordant (cf.\ \cite{yodzis72}) in the sense that (a) there is an open connected set $O$ such that $\p O=S_1\cup S_2$, and (b) there is a smooth future-directed  timelike vector field $V$ which  points to the interior of $\bar{O}$ on $S_1$ and to the exterior of $\bar{O}$ on $S_2$.

Moreover, suppose that
there is an open set $G\subset O$ such that the  1-flow map of $V$ establishes a diffeomorphism between $S_1\backslash K$ and $S_2\backslash K$, $K=\bar{G}$ (see Fig.\ \ref{kxs}; we do not assume that $K$ is compact).
Furthermore, assume that $K\cap S_1$ is compact and
\begin{itemize}
\item[(*)] there is a set $C$ generated by integral segments of $V$ such that the boundary of $G$ is a disjoint union $\p G= (K\cap S_1)\cup (K\cap S_2)\cup C$ and $C\subset \textrm{Int} D^+(S_1)$.
 \end{itemize}
Then one of the following possibilities holds:
\begin{itemize}
\item[(i)] $K\subset D^+(S_1)$, every integral curve of $V$ starting from $S_1$ reaches $S_2$, $\bar{O}=S_1\times [0,1]$, $K$ is compact, $S_1$ and $S_2$ are diffeomorphic.
\item[(ii)] $K$ is non compact and contains an integral curve of $V$ which escapes every compact set in some direction and intersects at most one set among  $S_1\cap K$ and $S_2\cap K$.
\item[(iii)] $H^+(S_1)\subset G$ is compact and non-empty, $S_1$ is compact and has zero Euler characteristic and: $S_1$ is diffeomorphic to $S_2$ or there is a closed timelike curve inside $K$.
\end{itemize}
\end{theorem}

\begin{figure}[ht!]
\centering
%\psfrag{A}{$S_2$} \psfrag{B}{$S_1$} \psfrag{c}{$\!\!\!C$} \psfrag{G}{$G$} \psfrag{h}{$\!\!V$}
\includegraphics[width=9.5cm]{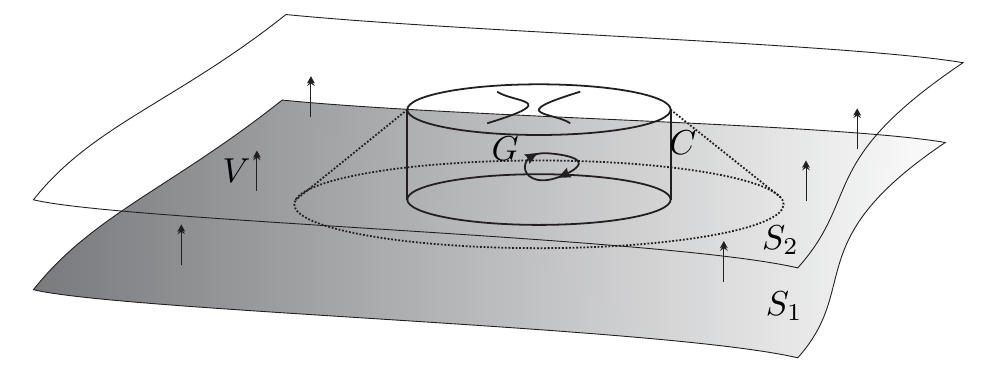}
\caption{The geometrical elements of Theorem \ref{laz}. The set $O$ is the open set between $S_1$ and $S_2$. The pathological behavior is inside $G$. The boundary $C$ is such that $J^{-}(C)\cap S_1$ is compact.} \label{kxs}
\end{figure}

Recall that by definition the empty set is compact.
The condition (*)  is weaker than the condition used by Geroch and Tipler for a similar purpose \cite{tipler77}, namely $\bar{O}\backslash K\subset D^{+}(S_1)$. It is trivially satisfied if $S_1$ is compact as in Georch's theorem \cite[Theor.\ 2]{geroch67}, just take $G=O$ so that $C=\emptyset$.
It is also satisfied if $M$ is  weakly asymptotically simple and empty (WASE) and asymptotically predictable  \cite{hawking73}. It serves to separate the nasty spacetime behavior on $G$ from the nasty spacetime behavior that may happen at spacelike infinity (in this sense Geroch and Tipler's condition demands that there is no nasty behavior at spacelike infinity). This condition makes sense since we  want to focus on the possibility of topology change or chronology violation caused by the evolution of spacetime and not on pathologies already present at spacelike infinity.

The theorem proves that if the partial Cauchy hypersurface  $S_1$ is non-compact or compact with Euler characteristic different from zero, so that (iii) does not apply, then the evolution does not involve neither chronology violation nor topology change (i.e.\ (i) applies) unless the spacetime continuum is broken (i.e.\ (ii) holds).
Since in three space dimensions any closed  manifold $S_1$ has  zero Euler characteristic we can conclude that physically,
without breaking the spacetime continuum it is impossible, in classical general relativity, to locally create time machines or to locally change the topology of space. If these pathologies take place they involve the whole Universe (which must have compact section).

If the spacetime continuum is not broken we are in the  compact cobordism case, which is  obtained as a corollary setting $O=G$ and $C=\emptyset$.

\begin{theorem}
Let $(M,g)$ be a spacetime which satisfies the null convergence condition.
Suppose that two disjoint $C^1$ spacelike hypersurfaces are compactly Lorentz cobordant  in the sense that (a) there is an open connected relatively compact set $O$ such that $\p O=S_1\cup S_2$, and (b) there is a smooth future-directed  timelike vector field $V$ which  points to the interior of $\bar{O}$ on $S_1$ and to the exterior of $\bar{O}$ on $S_2$.
Then one of the following mutually excluding possibilities holds.
\begin{enumerate}
\item $\bar{O}\subset D^+(S_1)$, every integral curve of $V$ starting from $S_1$ reaches $S_2$ and conversely, $\bar{O}=S_1\times [0,1]$, $S_1$ and $S_2$ are diffeomorphic.
\item $H^+(S_1)\subset O$ is compact and non-empty (so all the special properties of Theorem \ref{mai} apply), $S_1$  has zero Euler characteristic and: $S_1$ is diffeomorphic to $S_2$ or there is a closed timelike curve inside $\bar{O}$.
\end{enumerate}
\end{theorem}

\begin{proof}[of Theor.\ \ref{laz}]
First let us show that if  $D^+(S_1)\cap K$ is non-compact then (ii) applies. Take a sequence $q_n \in   D^+(S_1)\cap K$  escaping every compact set, and consider the integral curves of $V$ which start from some point $p_n\in S_1\cap K$ and end at $q_n$. Passing to the limit we find $p\in S_1\cap K$ and a limit curve starting from $p$ (cf.\ \cite{beem96,minguzzi07c}), which is an integral curve of $V$, which escapes every compact set. Thus we can assume without loss of generality that $D^+(S_1)\cap K$ is compact.

Suppose that $K\subset D^+(S_1)$ and hence that $K$ is compact. Since $S_2$ is spacelike and $C\subset \textrm{Int} D^+(S_1)$ we have also $K\backslash S_1\subset \textrm{Int} D^+(S_1)$. Every integral curve of $V$ which passes through some point $q\in K\cap S_2$, once extended to the past, must intersect $S_1$. Conversely, every integral curve of $V$ which passes through some point $p\in K\cap S_1$ once extended to the future must intersect $K\cap S_2$, for otherwise it would be imprisoned in a compact set $K$ where strong causality holds (because $K\backslash S_1\subset \textrm{Int} D^+(S_1)$), a contradiction. Thus the 1-flow map of a suitable vector field obtained rescaling $V$ provides a diffeomorphism between $S_1$ and $S_2$. The same flow provides a homotopy which shows that $\bar{O}=S_1\times [0,1]$ and $K=(S_1\cap K)\times [0,1]$. Thus (i) applies.

Observe that $K\subset D^+(S_1)$ if and only if $K\cap H^+(S_1)=\emptyset$.
One direction has been already shown, for the other direction, since $O$ is connected it is path connected. For any $q\in K$ there is a curve $c: [0,1]\to K\backslash S_1$ such that $c(0)\in C\cup (S_1\cap K)$, $c(1)=q$. But this curve cannot cross $H^+(S_1)$ so the whole curve remains in $\textrm{Int} D^+(S_1) \cup (S_1\cap K)$, hence $K\subset D^+(S_1)$.

It remains to consider the case in which $K\cap H^+(S_1)$ is non-empty and $K  \cap D^+(S_1)$ (and $K\cap H^+(S_1)$) is compact. Since $S_2$ is spacelike $H^+(S_1)\cap G$ is non-empty.
No generator of $H^+(S_1)$ can intersect $C$ otherwise by (*) it would be forced to reach $S_1$ which is impossible since $S_1$ is edgeless.
Let $U=I^{-}(O)$ then on the spacetime $(U,g\vert_U)$, $H^+_U(S_1)\cap G\ne \emptyset$, and $H=H^+_U(S_1)\cap K$ is a compactly generated component of $H_U^+(S_1)$ contained in $G$.  Since $H$ is edgeless and does not intersect $C$ the vector field $V$ induces a homeomorphism between $H$ and a connected compact component $A$ of $S_1\cap K$ which does not intersect $C$. But $H^+(A)=H$, and since $O$ is connected any curve connecting $C$ to $A$ must first escape $\textrm{Int} D^+_U(S_1)$ without entering $D^+(A)$ which is impossible since $H=H^+_U(S_1)\cap K=H^+(A)$. This argument and Theorem \ref{mai} prove that $C=\emptyset$, $S_1$ is compact, $H=H_U^+(S_1)=H^+(S_1)$, it is $C^2$, and homeomorphic with $S_1$.
As $H$ admits a global $C^1$ lightlike vector field, it has vanishing Euler characteristic and so $S_1$ has vanishing Euler characteristic.

If every integral curve of $V$ which intersects $S_1$ intersects $S_2$  and conversely then $S_1$ and $S_2$ are diffeomorphic by the usual 1-flow map of a vector field obtained rescaling $V$. If, on the contrary, there is an integral curve $\gamma \subset \bar{O}$ which does not intersect both $S_1$ and $S_2$ then either it escapes every compact set in some direction,
so (ii) applies and we have finished, or it is future imprisoned in a compact set
and hence accumulates over a point $p\in K$. Since $\gamma$ is  an integral curve of $V$, it  accumulates over a small integral open segment of $V$ passing through $p$. As a consequence, up to an initial segment of $\gamma$, the curve $\gamma$ is contained in both $I^+(p)$ and $I^{-}(p)$ which proves that $p\ll p$, that is there is a closed timelike curve $\eta$ in $K$ passing through $p$ (this curve  cannot pass through $S_1$ or $S_2$ for, as they are spacelike hypersurfaces, it would not be able to return to $p$). (this last argument is similar to that used by Geroch in \cite[Theor.\ 2]{geroch67}). $\square$
\end{proof}

Another theorem on topology change and compact Cauchy horizons can be found in a paper by Chru\'sciel and Isenberg \cite{chrusciel93} where they relied on differentiability conditions on the horizon. If the null convergence condition is added so as to assure the $C^2$ differentiability of the horizon, as proved in this work, their theorem gets included in Theorem \ref{mai}.

Theorem \ref{laz} does not contain causality assumption as it contemplates the possibility of formation of closed timelike curves. The assumption of global hyperbolicity removes altogether the possibility of topology change as any two Cauchy hypersurfaces are diffeomorphic. Still one can obtain relevant result of topological nature. Under global hyperbolicity and asymptotic flatness
 Gannon's singularity theorem \cite{gannon75} establishes that such spacetime would develop singularities, while the so called topological censorship  theorem  \cite{friedman93,galloway95} establishes that those singularities could not be probed by an observer free (to have come from and)
 to go to null infinity.

\subsection{Black holes}
We say that $(\bar{M},\bar{g})$ is a conformal completion of $(M,g)$ if there is a Lorentzian manifold with boundary $\bar{M}$ with interior $M$,  and a $C^3$ function $\Omega\colon \bar{M}\to[0,+\infty)$, such that $\p M=\bar{M}\backslash M=\{p\in \bar{M}: \Omega(p)=0\}$,  $\dd \Omega\ne 0$ on $\p M$, and $\Omega^2 g=\tilde{g}$ on $M$. The condition $\dd \Omega\ne 0$ tells us that $\p M$ is an embedded manifold on any local extension of $\bar{M}$ (regular level set theorem  \cite{lee03}).

We stress that  $\bar{M}$ is not necessarily compact, thus there can  be causal curves on $M$ escaping every compact set which do not have endpoint at $\bar{M}$.
We do not assume that a neighborhood of $M$ at infinity  is Ricci flat (empty). This condition would imply that $\p M$ is lightlike \cite{hawking73} a condition which we also do not assume. We are adopting the broad framework introduced in \cite[Sect.\ 4]{chrusciel01}. Let $\mathscr{I}^+=I^+(M;\bar{M})\cap \p M$.

\begin{definition}
The {\em event horizon} is the set $H=\p I^{-}(\mathscr{I}^+;\bar{M})\cap M$.
\end{definition}

The set $H$ is the boundary of a past set in $\bar{M}$, and as such it is locally Lipschitz, achronal and generated by future inextendible lightlike geodesics, hence it is a past horizon according to our definition.

Let $\gamma$ be a lightlike geodesic with future endpoint in $\mathscr{I}^+$. The affine parameters of $\gamma$ with respect to $g$ and $\tilde{g}$ on $M$ are related by $\dd v=\Omega^{-2} \dd \tilde v$.
Since $\Omega=0$, $\dd \Omega\ne 0$ at $\p M$, it is easy to show, Taylor expanding $\Omega$ at $\gamma\cap \p M$,  that $\gamma$ is future complete.

Thus if we could show that every future generator of $H$ reaches $\mathscr{I}^+$
then we would have by Theorem \ref{pot} that $\theta\ge 0$ on $H$, which would imply that the area  of the horizon section increases according to Theorem \ref{are}. Unfortunately, no generator  $\gamma$ of $H$ starting from $p\in H$ can have future endpoint at $\mathscr{I}^+$, indeed, regardless of the causal type of this hypersurface, it would be possible to deform $\gamma$ into a timelike curve connecting $p$ to a some point of  $\mathscr{I}^+$ in contradiction with the definition of $H$.

The correct argument as already conceived by Hawking \cite[Lemma 9.2.2]{hawking73} is slightly more complex and involves a lightlike geodesic running near the horizon rather than on the horizon. Here one has to impose a condition which assures that the lightlike geodesic will eventually reach $\p M$ so as to take advantage of its future completeness. This condition can be \cite{chrusciel01}
\begin{itemize}
\item[$\star$] There is a neighborhood $O$ of $H$ such that for every compact set $C\subset O$, $C\cap I^{-}(\mathscr{I}^+,\bar{M})\ne \emptyset$, there is a future inextendible (in $M$) geodesic $\eta \subset \p J^+(C,{M})$ with future endpoint at $\mathscr{I}^+$.
\end{itemize}
Since $\eta$ is the generator of a future set it is an achronal lightlike geodesic.
 The reader will not find this assumption in  Hawking's work since he deduces it from stronger but  physically motivated conditions on the asymptotic structure, and in particular from the assumption of {\em asymptotic predictability
} (weak cosmic censorship): $\mathscr{I}^+\subset \overline{D^+(S)}$ where $S$ is a partial Cauchy hypersurface and the closure is in the topology of $\bar{M}$. The reader is referred to \cite{hawking73,chrusciel01} for a discussion of the reasonability of $\star$.

The fact that Hawking's argument on the positivity of $\theta$ can be adapted to the non-smooth case is  non-trivial and has been proved in \cite[Theor.\ 4.1]{chrusciel01}. For completeness and for the reader convenience we sketch the proof.

\begin{theorem} \label{pdi}
Let $H$ be the event horizon in a spacetime $(M,g)$ which satisfies the null convergence condition. Suppose that $\star$ holds true, then $\theta\ge 0$ on $H$.
\end{theorem}

\begin{proof}[Sketch]
By contradiction , suppose that there is an Alexandrov point $p\in H$  for which $\theta(p)<0$. Let $W$ be a local timelike hypersurface passing through $p$ generated by the smooth timelike vector field $V$ which we used in the local description of Sect.\ \ref{kdp}. In the coordinate statement given below we use the coordinates introduced there. The point $p\in H$ is also an Alexandrov point for $S:=H\cap W$. Let $S_n$ be a sequence of smooth codimension 2 manifold on $W$ approximating $S$. We denote with $p_n$ the unique point of $S_n$ such that the flow of $V$ sends $p_n$ to $p$. The sequence $S_n$ is built in such a way that (a) $x^0(p)-x^0(p_n)=1/n^2$, (b) the Hessian $a_n(p)$ of the graph function of $S_n$ at $p$ satisfies $a_n(p)=a+\frac{1}{n} I$. Since $\theta(p)<0$, for sufficiently large $n$ we have $\theta_n(p)<0$ where $\theta_n$ is the expansion of the lightlike congruence contained in $\p J^+(S_n)$.
Reducing the manifold $S_n$ if necessary we can assume that $\theta_n<0$ on $S_n$.
Let $O$ be the open subset defined by $\star$. We have for sufficiently large $n$,  $S_n\subset O$ and $\p S_n\subset I^+(S)$. Let $n$ be one such large value. Let us consider the compact set $\bar{S}_n$, since from (a) $S_n\cap I^{-}(\mathscr{I}^+, \bar{M})\ne \emptyset$ we have by $\star$ that there is a lightlike geodesic $\eta$ connecting some $q\in \bar{S}_n$ to $\mathscr{I}^+$. Now, $q$ cannot belong to $\p S_n$ otherwise $\p S_n\subset I^+(S)$ would give that $S$ and hence $H$ has some point in the chronological past of $\mathscr{I}^+$. Thus $q\in S_n$ and since $\theta(q)<0$ and $\eta$ is complete we get a contradiction. $\square$
\end{proof}

From Theorem \ref{are} we have

\begin{theorem}(Area theorem for event horizons) \label{kdd}
Let $H$ be an event horizon in a spacetime $(M,g)$ which satisfies the null convergence condition. Let $\tau\colon M \to \mathbb{R}$ be a Lipschitz time function on $M$.  Let $D(t)=\{\tau^{-1}((-\infty,t))\}\cap H$ and let $A(t):=\vert A(\p_+D(t),1)\vert$ be the area of the horizon, which is well defined for almost every $t$.
If $\star$ holds then for almost every $t_1, t_2\in \tau(M)$, $t_1< t_2$, we have
\[
A(t_1)\le A(t_2),
\]
 where if the equality holds then $\tau^{-1}((t_1,t_2))\cap H$ has at least the same regularity as the metric ($C^3$). In particular, if the topologies of $\tau^{-1}(t_1)\cap H$ and $\tau^{-1}(t_2)\cap H$ differ, for instance  if they have a different number of components, then the inequality is strict.
\end{theorem}

\begin{remark}
By Theorem \ref{ljs} one could apply the divergence theorem for every choice of $t_1<t_2$. The `almost every' restriction is due to the fact that
the divergence theorem involves $A(\p_- (H\backslash D(t_1)),1)\vert$ while we want to use $A(t_1)$ in its place, and they are equal only for almost every $t_1<t_2$.
\end{remark}

\begin{proof}
The inequality  follows from the just given argument and from the comment after Prop.\ \ref{alp}. Suppose that the equality holds, then by the area theorem \ref{are} and the inequality $\theta\ge 0$ (Theor.\ \ref{pdi}) we have that locally $\mu_{ij}^s=0$ and $\theta=0$ which implies by Theor.\ \ref{app} that the horizons has the same regularity as the metric.

The last statement follows from this observation: if the horizon is $C^2$ on an open set then the generators have no endpoint there, thus the flow of $n$ is well defined between times $t_1$ and $t_2$ and provides a homeomorphism between $H\cap \tau^{-1}(t_1)$ and $H\cap \tau^{-1}(t_2)$ thus they have the same topology. As a consequence, if the topologies of these slices differ then $\tau^{-1}((t_1,t_2))\cap H$ cannot be $C^2$ which by Theorem \ref{are} implies that the inequality is strict. $\square$
\end{proof}

This result was obtained by Hawking under tacit differentiability assumptions on the horizon \cite[Prop.\ 9.2.7]{hawking73}, and then generalized to the non-differentiable case in \cite{chrusciel01}. It is also  generically referred as the {\em area theorem}.
By Theorem \ref{are} the existence of a time function is inessential and serves only to identify the event horizon slices.

The following result has been regarded as a simple corollary of the area theorem, but can in fact be proved without imposing the null convergence condition (compare with \cite[Theor.\ 4.11]{chrusciel08}).

\begin{theorem}
Let $H$ be a horizon generated by a   lightlike Killing field which is nowhere vanishing  on $H$, then the horizon has the same regularity as the metric.
\end{theorem}

\begin{proof}
By assumption $H$ is sent into itself by the local flow of $k$, which means that through each point of $H$ passes an integral curve of $k$, necessarily an achronal lightlike geodesic, hence $H$ is $C^1$. Observe that $k$ is a semitangent field on $H$, that is we can set $n:=k$, and that by the Killing condition of $k$ on $H$ the expansion $\theta$ vanishes. Thus by Theorem \ref{app} the horizon is as regular as the metric. $\square$
\end{proof}

\subsection{Cosmic censorship and horizon rigidity}
It is expected that generically the maximal globally hyperbolic Cauchy development of matter and gravitational fields starting from appropriate Cauchy data on a spacelike hypersurface should  lead to a spacetime which cannot be further extended. A precise definition of this hypothesis, termed {\em strong cosmic censorship conjecture}, will not be particularly important for our purposes.  Since the conjecture asks to prove that generically horizons do not form one could try to prove, to start with, that generically compact Cauchy horizons do not form.

We are now going to prove this result which we state in the past version, although the physical  interesting case is the dual  future version. In our terminology the Einstein equations might or might not include a cosmological constant.

We recall that the {\em dominant energy condition} states that at each event $p$ the endomorphism of $T_pM$, $u^\alpha \to -T^\alpha_{\ \beta} u^\beta$, sends the future non-spacelike cone into itself. The {\em stable dominant energy condition} states that the endomorphism sends the future-directed causal cone into the future-directed timelike cone. It excludes forms of matter that are on the verge of violating the dominant energy condition and, in particular, some aligned pure radiation stress-energy tensor (Type II, \cite{hawking73}).

However, as the energy condition is a condition on the nature of the source, it is reasonable to  demand the stable dominant energy condition only if some source is present, i.e.\ $T\ne 0$.
This observation leads us to the  {\em weakened stable dominant energy condition} which states that the stable dominant energy condition holds  wherever $T\ne 0$. It requires that the above  endomorphism sends the future-directed causal cone into the future-directed timelike cone plus the zero vector \cite{minguzzi14e}. In the physical four dimensional case it  allows diagonal stress energy tensors in which the energy density is larger than the absolute value of the principal pressures.

\begin{theorem}  \label{aoa}
Suppose that the Einstein equations hold on $(M,g)$.
Let  $S$ be a $C^1$ connected  partial Cauchy hypersurface and suppose that (i) the weakened stable dominant energy condition holds, and (ii) $T\ne 0$ somewhere on $S$, then all the components of $H^{-}(S)$ are neither compact nor compactly generated.
\end{theorem}

Condition (ii) states that there is some form of energy content on spacetime, that is, spacetime is not empty. It can also be regarded as a kind of genericity condition.

\begin{proof}
Suppose that $H^{-}(S)$ has a compactly generated component, then by Theorem \ref{mai} $H^{-}(S)$ has just one compact $C^2$ component and $S$ is compact. By the weakened stable dominant energy condition and by the conservation theorem as clarified and improved in \cite[Prop.\ 3.5]{minguzzi14e} $T(n,n)\ne 0$ somewhere on $H^{-}(S)$ where $n$ is the semitangent to the horizon, which is impossible because $T(n,n)=R(n,n)$ and by Theorem \ref{mai}, $R(n,n)=0$ on the horizon. $\square$
\end{proof}

\begin{corollary}
Let $(M,g)$ be a spacetime which satisfies the Einstein equations and such that the weakened stable dominant energy condition holds.
Let  $S$ be a connected  partial Cauchy hypersurface such that $H^{-}(S)$ is compact, then the stress-energy tensor vanishes on $\overline{D(S)}$, so the vacuum Einstein equations hold on it, and $S$ is compact and with zero Euler characteristic.
\end{corollary}

\begin{proof}
 By Theorem \ref{aoa} and the  energy condition the stress energy tensor vanishes on $S$, thus by Hawking's conservation theorem  \cite{hawking73} as improved and clarified in \cite{minguzzi14e} the stress energy tensor vanishes on $\overline{D(S)}$. The last statement follows from Theorem \ref{mai}. $\square$
\end{proof}

In the empty case there is still the possibility that a compact Cauchy horizon could form. However, one would expect that this could occur only in very special (non generic) cases as in the highly symmetric Taub-NUT solution. A very interesting result in this direction is due to Moncrief and Isenberg who showed that any analytic compact Cauchy horizon generated by closed lightlike geodesics is actually generated by a lightlike Killing field \cite{moncrief83,isenberg85}. The analyticity condition was subsequently improved to smoothness by Friedrich, R\'acz and Wald \cite{friedrich99}. Joining their main theorem with our smoothness result we obtain

\begin{theorem}
Let $S$ be a compact Cauchy hypersurface in an electro-vacuum smooth spacetime. Then $H^{-}(S)$ if non-empty is smooth, and if its generators are past incomplete and closed lightlike geodesics then there is a neighborhood $U$ of the horizon such that on $J^{+}(H)\cap U$ there is a smooth Killing field which is normal to $H$.
\end{theorem}

Under analyticity a similar result holds true, but the Killing field exists all over $U$.
The problem of removing the condition on the closure of the geodesics remains open.

\section{Conclusions}

We have obtained and improved some known results on  the differentiability of horizons giving new and simple proofs based on just its semi-convexity properties.
Then we have reviewed and improved the area theorem offering a novel approach based on the divergence theorem for divergence measure fields. The new version can be applied to a wider family of domains and relates the area increase with the integral of the divergence. The equality case has been studied in detail showing that it corresponds to the vanishing of the singular part of the divergence (or of the Hessian of the horizon graphing function).

The application of some regularity results on quasi-linear elliptic PDEs has lead us to the proof that under the null energy condition every compactly generated Cauchy horizon is smooth and compact, thus solving a  known  open problem in mathematical relativity.

Finally, these results have been applied to different more specific issues: (1) we obtained the first complete proof of Hawking's theorem on the (classical) non-existence of time machines, (2) we obtained some other theorems which showed that an advanced civilization cannot create regions of topology change without breaking the spacetime continuum. These theorems do not use the genericity condition and show that the formation of closed timelike curves do not spoil the conclusion.  (3) We showed how to apply our version of the  area theorem to obtain some classical results on the smoothness of event horizons, and on the increase of the black hole area under merging, (4) we showed that under the weakened stable dominant energy condition and for universes with some energy content, compact Cauchy horizons do not form, a result which supports the strong cosmic censorship. Further, our smoothness result allows us to remove a relevant assumption in the classical theorem by Moncrief and Isenberg on the Killing properties of compact Cauchy horizons.

\section*{Acknowledgments}   This work has been
partially supported by GNFM of INDAM.

\subsection*{A comment on a similar work}
$\empty$\\
This work, without Sect.\ \ref{awn} and Theor.\ \ref{lla}, was posted on the Archive (arXiv: 1406.5919) as the last of a series of three papers (the others being \cite{minguzzi14,minguzzi14e}). The very next day a related work by E.\ Larsson    (arXiv:1406.6194 recently published in \cite{larsson14}), reaching similar conclusions on smoothness of Cauchy horizons and topology change, was also posted (\cite{larsson14} mentions that it also appeared some days before on a public web repository of theses of the KTH Institute, Stockholm). This work and Larsson's follow
quite different lines of proof. This one uses some results from geometric measure theory to prove and strengthen the area theorem, while his develops a strategy based on a flow over the horizon and relies on work initiated in \cite{chrusciel01}.

\def\cprime{$'$}


\begin{thebibliography}{10}

\bibitem{alberti94}
Alberti, G.: On the structure of singular sets of convex functions.
\newblock Calc. Var. Partial Differential Equations \textbf{2}, 17--27 (1994)

\bibitem{alberti92}
Alberti, G., Ambrosio, L., and Cannarsa, P.: On the singularities of convex
  functions.
\newblock Manuscripta Math. \textbf{76}, 421--435 (1992)

\bibitem{alberti13}
Alberti, G., Bianchini, S., and Crippa, G.: Structure of level sets and
  {S}ard-type properties of {L}ipschitz maps.
\newblock Ann. Sc. Norm. Super. Pisa Cl. Sci. \textbf{12}, 863--902 (2013)

\bibitem{ambrosio00}
Ambrosio, L., Fusco, N., and Pallara, D.: \emph{Functions of bounded variation
  and free discontinuity problems}.
\newblock Oxford: Claredon {P}ress (2000)

\bibitem{andersson98}
Andersson, L., Galloway, G.~J., and Howard, R.: The cosmological time function.
\newblock Class. Quantum Grav. \textbf{15}, 309--322 (1998)

\bibitem{beem96}
Beem, J.~K., Ehrlich, P.~E., and Easley, K.~L.: \emph{Global Lorentzian
  Geometry}.
\newblock New York: Marcel {D}ekker {I}nc. (1996)

\bibitem{beem98}
Beem, J.~K. and Kr{\'o}lak, A.: Cauchy horizon end points and
  differentiability.
\newblock J. Math. Phys. \textbf{39}, 6001--6010 (1998)

\bibitem{bianchi96}
Bianchi, G., Colesanti, A., and Pucci, C.: On the second differentiability of
  convex surfaces.
\newblock Geom. Dedicata \textbf{60}, 39--48 (1996)

\bibitem{borde84}
Borde, A.: A note on compact {C}auchy horizons.
\newblock Phys. Lett. A \textbf{102}, 224--226 (1984)

\bibitem{borde97}
Borde, A.: How impossible is topology change?
\newblock Bull. Astr. Soc. India \textbf{25}, 571--577 (1997)

\bibitem{borde94b}
Borde, A.: Topology change in classical general relativity (2004).
\newblock arXiv:gr-qc/9406053

\bibitem{budzynski01}
Budzy{\'n}ski, R., Kondracki, W., and Kr{\'o}lak, A.: New properties of
  {C}auchy and event horizons.
\newblock Nonlinear Analysis \textbf{47}, 2983--2993 (2001)

\bibitem{budzynski03}
Budzy{\'n}ski, R., Kondracki, W., and Kr{\'o}lak, A.: On the differentiability
  of compact {C}auchy horizons.
\newblock Lett. Math. Phys. \textbf{63}, 1--4 (2003)

\bibitem{budzynski99}
Budzy{\'n}ski, R.~J., Kondracki, W., and Kr{\'o}lak, A.: On the
  differentiability of {C}auchy horizons.
\newblock J. Math. Phys. \textbf{40}, 5138--5142 (1999)

\bibitem{cannarsa04}
Cannarsa, P. and Sinestrari, C.: \emph{Semiconcave functions,
  {H}amilton-{J}acobi equations, and optimal control}.
\newblock Progress in Nonlinear Differential Equations and their Applications,
  58. Birkh\"auser Boston, Inc., Boston, MA (2004)

\bibitem{chen01}
Chen, G.-Q. and Frid, H.: On the theory of divergence-measure fields and its
  applications.
\newblock Bol. Soc. Bras. Mat. \textbf{32}, 401--433 (2001)

\bibitem{chen05}
Chen, G.-Q. and Torres, M.: Divergence-measure fields, sets of finite
  perimeter, and conservation laws.
\newblock Arch. Ration. Mech. Anal. \textbf{175}, 245--267 (2005)

\bibitem{chen09}
Chen, G.-Q., Torres, M., and Ziemer, W.~P.: Gauss-{G}reen theorem for weakly
  differentiable vector fields, sets of finite perimeter, and balance laws.
\newblock Comm. Pure Appl. Math. \textbf{62}, 242--304 (2009)

\bibitem{chrusciel98b}
Chru{\'s}ciel, P.~T.: A remark on differentiability of {C}auchy horizons.
\newblock Class. Quantum Grav. \textbf{15}, 3845–--3848 (1998)

\bibitem{chrusciel08}
Chru{\'s}ciel, P.~T. and Costa, J.~L.: On uniqueness of stationary vacuum black
  holes.
\newblock Ast\'erisque pages 195--265 (2008).
\newblock G{\'e}om{\'e}trie diff{\'e}rentielle, physique math{\'e}matique,
  math{\'e}matiques et soci{\'e}t{\'e}. I

\bibitem{chrusciel01}
Chru{\'s}ciel, P.~T., Delay, E., Galloway, G.~J., and Howard, R.: Regularity of
  horizons and the area theorem.
\newblock Ann. {H}enri {P}oincar\'e \textbf{2}, 109–--178 (2001)

\bibitem{chrusciel02}
Chru{\'s}ciel, P.~T., Fu, J. H.~G., Galloway, G.~J., and Howard, R.: On fine
  differentiability properties of horizons and applications to {R}iemannian
  geometry.
\newblock J. Geom. Phys. \textbf{41}, 1--12 (2002)

\bibitem{chrusciel98}
Chru{\'s}ciel, P.~T. and Galloway, G.~J.: Horizons non-differentiable on a
  dense set.
\newblock Commun. Math. Phys. \textbf{193}, 449–--470 (1998)

\bibitem{chrusciel93}
Chru{\'s}ciel, P.~T. and Isenberg, J.: Compact {C}auchy horizons and {C}auchy
  surfaces.
\newblock In e.~Jacobson, editor, \emph{Directions in General Relativity (Brill
  Festshrift)}. Cambridge: CUP (1993), vol.~2, pages 97--107

\bibitem{chrusciel94}
Chru{\'s}ciel, P.~T. and Isenberg, J.: On the dynamics of generators of
  {C}auchy horizons.
\newblock In A.~B. D.~Hobill and e.~A.~Coley, editors, \emph{Proceedings of the
  Kananaskis conference on chaos in general relativity}. Plenum (1994), pages
  113--125.
\newblock MPA preprint MPA 773

\bibitem{clarke75}
Clarke, F.~H.: Generalized gradients and their applications.
\newblock Trans. Amer. Math. Soc. \textbf{205}, 247--262 (1975)

\bibitem{clarke95}
Clarke, F.~H., Stern, R.~J., and Wolenski, P.~R.: Proximal smoothness and the
  lower-{$C^2$} property.
\newblock {J}. {C}onvex Anal. \textbf{2}, 117--144 (1995)

\bibitem{colombo06}
Colombo, G. and Marigonda, A.: Differentiability properties for a class of
  non-convex functions.
\newblock Calc. Var. Partial Differential Equations \textbf{25}, 1--31 (2006)

\bibitem{daniilidis05}
Daniilidis, A. and Malick, J.: Filling the gap between lower-$C^1$ and lower-$C^2$ functions.
\newblock {J}. {C}onvex Anal. \textbf{12}, 315--329 (2005)

\bibitem{dudley77}
Dudley, R.~M.: On second derivatives of convex functions.
\newblock {M}ath. Scand. \textbf{41}, 159--174 (1977)

\bibitem{evans98}
Evans, L.~C.: \emph{Partial differential equations}.
\newblock Providence: American {M}athematical {S}ociety (1998)

\bibitem{evans92}
Evans, L.~C. and Gariepy, R.~F.: \emph{Measure theory and fine properties of
  functions}.
\newblock Boca {R}aton: {CRC} {P}ress (1992)

\bibitem{federer59}
Federer, H.: Curvature measures.
\newblock Trans. Amer. Math. Soc. \textbf{93}, 418--491 (1959)

\bibitem{friedman93}
Friedman, J.~L., Schleich, K., and Witt, D.~M.: Topological censorship.
\newblock Phys. Rev. Lett. \textbf{71}, 1486--1489 (1993)

\bibitem{friedrich99}
Friedrich, H., R{\'a}cz, I., and Wald, R.~M.: On the rigidity theorem for
  spacetimes with a stationary event horizon or a compact {C}auchy horizon.
\newblock Comm. Math. Phys. \textbf{204}, 691--707 (1999)

\bibitem{galloway95}
Galloway, G.~J.: On the topology of the domain of outer communication.
\newblock Class. Quantum Grav. \textbf{12}, L99--L101 (1995)

\bibitem{galloway00}
Galloway, G.~J.: Maximum principles for null hypersurfaces and null splitting
  theorems.
\newblock Ann. {H}enri {P}oincar\'e \textbf{1}, 543--567 (2000)

\bibitem{gannon75}
Gannon, D.: Singularities in nonsimply connected space-times.
\newblock J. Math. Phys. \textbf{16}, 2364--2367 (1975)

\bibitem{geroch67}
Geroch, R.: Topology in general relativity.
\newblock J. Math. Phys. \textbf{8}, 782--786 (1967)

\bibitem{giaquinta83}
Giaquinta, M.: \emph{Multiple integrals in the calculus of variations and
  nonlinear elliptic systems}, vol. 105 of \emph{Annals of Mathematics
  Studies}.
\newblock Princeton University Press, Princeton, NJ (1983)

\bibitem{giusti84}
Giusti, E.: \emph{Minimal surfaces and functions of bounded variation}.
\newblock Boston: Birkh{\"a}user (1984)

\bibitem{hartman64}
Hartman, P.: \emph{Ordinary differential equations}.
\newblock New York: John {W}iley {\&} Sons (1964)

\bibitem{hawking92}
Hawking, S.~W.: Chronology protection conjecture.
\newblock Phys. Rev. D \textbf{46}, 603--611 (1992)

\bibitem{hawking73}
Hawking, S.~W. and Ellis, G. F.~R.: \emph{The Large Scale Structure of
  Space-Time}.
\newblock Cambridge: Cambridge {U}niversity {P}ress (1973)

\bibitem{urruty85}
Hiriart-Urruty, J.-B.: \emph{Generalized Differentiability, Duality and
  Optimization for Problems Dealing with Differences of Convex Functions},
  Berlin: {Springer-Verlag}, vol. Convexity and Duality in Optimization of
  \emph{Lecture Notes in Economics and Mathematical Systems}, pages 37--70
  (1985)

\bibitem{hirsch76}
Hirsch, M.~W.: \emph{Differential topology}.
\newblock New York: {Springer-Verlag} (1976)

\bibitem{hofmann07}
Hofmann, S., Mitrea, M., and Taylor, M.: Geometric and transformational
  properties of {L}ipschitz domains, {S}emmes-{K}enig-{T}oro domains, and other
  classes of finite perimeter domains.
\newblock The Journal of Geometric Analysis \textbf{17}, 593--647 (2007)

\bibitem{isenberg85}
Isenberg, J. and Moncrief, V.: Symmetries of cosmological {C}auchy horizons
  with exceptional orbits.
\newblock J. Math. Phys. \textbf{26}, 1024--1027 (1985)

\bibitem{kar07}
Kar, S. and Sengupta, S.: The {R}aychaudhuri equations: {A} brief review.
\newblock {P}ramana – {J}. {P}hys. \textbf{69}, 49--76 (2007)

\bibitem{krasnikov14}
Krasnikov, S.: Yet another proof of {H}awking and {E}llis's {L}emma 8.5.5.
\newblock Class. Quantum Grav. \textbf{31}, 227001 (2014).
\newblock arXiv:1407.0340

\bibitem{krasnikov95}
Krasnikov, S.~V.: Topology change without any pathology.
\newblock Gen. Relativity Gravitation \textbf{27}, 529--536 (1995)

\bibitem{kupeli87}
Kupeli, D.~N.: On null submanifolds in spacetimes.
\newblock Geom. Dedicata \textbf{23}, 33--51 (1987)

\bibitem{ladyzhenskaya68}
Lady{\v{z}}enskaja, O.~A. and Ural{\cprime}tseva, N.~N.: \emph{Linear and
  quasilinear elliptic equations}.
\newblock Academic Press, New York (1968)

\bibitem{larsson14}
Larsson, E.: Smoothness of compact horizons.
\newblock Ann. Henri Poincar\'e. DOI:10.1007/s00023-014-0371-z,  arXiv:1406.6194

\bibitem{lee03}
Lee, J.~M.: \emph{Introduction to smooth manifolds}.
\newblock New York: {Springer-Verlag} (2003)

\bibitem{mikusinski78}
Mikusi\'nski, J.: \emph{The Bochner integral}.
\newblock New York: Academic Press (1978)

\bibitem{minguzzi07c}
Minguzzi, E.: Limit curve theorems in {L}orentzian geometry.
\newblock J. Math. Phys. \textbf{49}, {092501} (2008).
\newblock arXiv:0712.3942

\bibitem{minguzzi14}
Minguzzi, E.: Completeness of {C}auchy horizon generators.
\newblock J. Math. Phys. \textbf{55}, 082503 (2014).
\newblock arXiv:1406.5909

\bibitem{minguzzi14e}
Minguzzi, E.: The vacuum conservation theorem.
\newblock \emph{Gen. Relativ. Gravit.}, \textbf{47} (2015) 32.
\newblock arXiv:1406.5915

\bibitem{moncrief83}
Moncrief, V. and Isenberg, J.: Symmetries of cosmological {C}auchy horizons.
\newblock Comm. Math. Phys. \textbf{89}, 387--413 (1983)

\bibitem{morrey58}
Morrey, C.~B., Jr.: On the analyticity of the solutions of analytic non-linear
  elliptic systems of partial differential equations. {I}. {A}nalyticity in the
  interior.
\newblock Amer. J. Math. \textbf{80}, 198--218 (1958)

\bibitem{nijenhuis74}
Nijenhuis, A.: Strong derivatives and inverse mappings.
\newblock Amer. Math. Monthly \textbf{81}, 969--980 (1974)

\bibitem{pfeffer12}
Pfeffer, W.~F.: \emph{The divergence theorem and sets of finite perimeter}.
\newblock Boca Raton: {CRC} {P}ress (2012)

\bibitem{poisson04}
Poisson, E.: \emph{A relativist's toolkit}.
\newblock Cambridge: Cambridge {U}niversity {P}ress (2004)

\bibitem{reinhart63}
Reinhart, B.~L.: Cobordism and the {E}uler number.
\newblock Topology \textbf{2}, 173--177 (1963)

\bibitem{reshetnjak68}
Re{\v{s}}etnjak, J.~G.: Generalized derivatives and differentiability almost
  everywhere.
\newblock Mat. Sb. (N.S.) \textbf{75(117)}, 323--334 (1968)

\bibitem{rockafellar70}
Rockafellar, R.~T.: \emph{Convex Analysis}.
\newblock Princeton: Princeton University Press (1970)

\bibitem{rockafellar82}
Rockafellar, R.~T.: \emph{Favorable classes of Lipschitz continuous functions
  in subgradient optimization}, New York: Pergamon {P}ress, vol. Progress in
  Nondifferentiable Optimization, pages 125--144 (1982).
\newblock IIASA Collaborative Proceedings Series, International Institute of
  Applied Systems Analysis, Laxenburg, Austria

\bibitem{rockafellar99}
Rockafellar, R.~T.: Second-order convex analysis.
\newblock Journal of Nonlinear and Convex Analysis \textbf{1}, 1--16 (1999)

\bibitem{rockafellar09}
Rockafellar, R.~T. and Wets, R. J.-B.: \emph{Variational analysis}.
\newblock Berlin: {Springer-Verlag} (2009)

\bibitem{schneider99}
Schneider, P., Ehlers, J., and Falco, E.~E.: \emph{Gravitational lenses}.
\newblock New York: {Springer} (1999)

\bibitem{silhavy05}
{\v S}ilhav{\'y}, M.: Divergence measure fields and {C}auchy's stress theorem.
\newblock Rend. Sem. Math. Univ. Padova \textbf{113}, 15--45 (2005)

\bibitem{tipler77}
Tipler, F.~J.: Singularities and causality violation.
\newblock Ann. {P}hys. \textbf{108}, 1--36 (1977)

\bibitem{tolksdorf84}
Tolksdorf, P.: Regularity for a more general class of quasilinear elliptic
  equations.
\newblock J. Differential Equations \textbf{51}, 126--150 (1984)

\bibitem{vial83}
Vial, J.-P.: Strong and weak convexity of sets and functions.
\newblock Math. Oper. Res. \textbf{8}, 231--259 (1983)

\bibitem{yodzis72}
Yodzis, P.: Lorentz cobordism.
\newblock Comm. Math. Phys. \textbf{26}, 39--52 (1972)

\bibitem{yodzis73}
Yodzis, P.: Lorentz cobordism. {II}.
\newblock Gen. Relativ. Gravit. \textbf{4}, 299--307 (1973)

\end{thebibliography}
\end{document}